\documentclass[11pt, letterpaper]{article}
\usepackage[letterpaper, margin=1in]{geometry}
\usepackage{amsthm, amsmath, amssymb}
\usepackage{mathtools}
\usepackage{graphicx} 
\usepackage[linesnumbered,ruled]{algorithm2e}
\usepackage{thm-restate}

\usepackage{float}

\usepackage{xcolor}

\usepackage{hyperref}

\title{Adversarially-Robust Gossip Algorithms for Approximate Quantile and Mean Computations}
\author{Bernhard Haeupler \thanks{Partially funded by the Ministry of Education and Science of Bulgaria's support for INSAIT as part of the Bulgarian National Roadmap for Research Infrastructure and through the European Research Council (ERC) under the European Union's Horizon 2020 research and innovation program (ERC grant agreement 949272).} \\
\texttt{bernhard.haeupler@inf.ethz.ch}\\
INSAIT, Sofia University ``St. Kliment Ohridski'' \& ETH Z\"urich\\ \and Marc Kaufmann \thanks{The authors gratefully acknowledge support by the Swiss National Science Foundation [grant number 200021\_192079].} \\
\texttt{marc.kaufmann@inf.ethz.ch}\\
ETH Z\"urich\\ \and Raghu Raman Ravi \thanks{The author gratefully acknowledges support by the Swiss National Science Foundation [grant number 10003390].}\\\texttt{raghu.ravi@inf.ethz.ch}\\ETH Z\"urich\\\and Ulysse Schaller \footnotemark[2] \\
\texttt{ulysse.schaller@inf.ethz.ch}\\ETH Z\"urich\\ }

\date{}

\newcommand{\R}{\mathbb{R}}

\DeclarePairedDelimiter\abs{\lvert}{\rvert}

\DeclarePairedDelimiter\ceil{\lceil}{\rceil}

\newcommand{\ex}[1]{\mathbb{E} \left[ #1 \right]}

\newcommand{\pr}[1]{\mathbb{P} \left[ #1 \right]}

\newtheorem{theorem}{Theorem}

\newtheorem{corollary}[theorem]{Corollary}

\newtheorem{lemma}[theorem]{Lemma}

\newtheorem{claim}[theorem]{Claim}

\begin{document}

\maketitle

\begin{abstract}
    \noindent This paper presents gossip algorithms for aggregation tasks that demonstrate both robustness to adversarial corruptions of any order of magnitude and optimality across a substantial range of these corruption levels.
     
    Gossip algorithms distribute information in a scalable and efficient way by having random pairs of nodes exchange small messages. Value aggregation problems are of particular interest in this setting, as they occur frequently in practice, and many elegant algorithms have been proposed for computing aggregates and statistics such as averages and quantiles. An important and well-studied advantage of gossip algorithms is their robustness to message delays, network churn, and unreliable message transmissions. However, these crucial robustness guarantees only hold if all nodes follow the protocol and no messages are corrupted.

    In this paper, we remedy this by providing a framework to model both adversarial participants and message corruptions in gossip-style communications by allowing an adversary to control a small fraction of the nodes or corrupt messages arbitrarily. Despite this very powerful and general corruption model, we show that robust gossip algorithms can be designed for many important aggregation problems. Our algorithms guarantee that almost all nodes converge to an approximately correct answer with optimal efficiency and essentially as fast as without corruptions.

    The design of adversarially-robust gossip algorithms poses completely new challenges. Despite this, our algorithms remain very simple variations of known non-robust algorithms with often only subtle changes to avoid non-compliant nodes gaining too much influence over outcomes. While our algorithms remain simple, their analysis is much more complex and often requires a completely different approach than the non-adversarial setting. 
\end{abstract}

\section{Introduction}

Distributed computing provides a resource-efficient framework for scalable algorithms that can handle the vast amounts of data produced by the increasingly interconnected and digitalized world. 
Designing fast and reliable algorithms to compute aggregates, such as averages, quantiles, maxima, or minima, is one active area of research in this field. These algorithms are especially vital in the context of sensor networks and P2P networks \cite{10.1145/1055558.1055597, 10.1145/844128.844142,  10.1145/1011767.1011809, 10.1145/1031495.1031524}. 

A promising class of distributed algorithms are gossip algorithms, which are typically simple, fast, scalable, and robust to errors. Due to this, many gossip algorithms have been developed for the aggregation of information in the real world. Their communication protocol is strikingly simple. In these algorithms, communication between nodes occurs in (synchronous) rounds - in each round, every node selects a communication partner uniformly at random from the set of all nodes and chooses to either push or pull a message from its partner. These messages are usually small -- $O(\log n)$ bits, where $n$ is the total number of nodes.

Push-based and pull-based protocols each have their own unique advantages and shortcomings. In \cite{kempe2003gossip}, the authors gave a simple, fast gossip algorithm to compute the mean and built on this to compute other aggregates, including arbitrary quantiles. The main algorithm to compute the mean is a simple yet elegant push-based protocol where each node pushes half of its current value to a random node in every round. Since such a push-protocol preserves the total sum of the values, it is not hard to see why this algorithm computes the mean, and it can be shown that it converges in $O(\log n)$ rounds with high probability. Improved algorithms that reduce the total number of messages sent were given in \cite{doi:10.1137/100793104} and \cite{10.1145/1142351.1142395}. In \cite{haeupler2018optimal}, the authors gave a faster, elegant, and in fact optimal algorithm to compute exact quantiles in $O(\log n)$ rounds with high probability, even in the presence of random faults. Their main algorithm is a pull-based tournament-style algorithm, where each node repeatedly updates its value to the median of the values of three randomly chosen nodes. This computes an approximate median extremely quickly - in $O(\log \log n + \log \frac{1}{\varepsilon})$ rounds for an $\varepsilon$-approximation - and can be bootstrapped to compute exact quantiles in $O(\log n)$ rounds.

One obstacle to the real-world deployment of these algorithms is the presence of message corruptions, be they malicious or benign. In the case of some algorithms, such as the median finding algorithm in \cite{haeupler2018optimal}, they are shown to be robust to random failures, that is, where nodes fail to perform their (pull- or push-) operation. The much more challenging problem, however, lies in the design of algorithms that are also robust to - potentially adversarial - message corruption. This arguably also greatly increases their applicability in the real world.
Indeed, some algorithms break completely in the presence of even a few adversarial faults in each round. This is the case for the push-protocol in \cite{kempe2003gossip}, where if even a single node reports a value that is abnormally large, all nodes would converge to an average that is much larger than the true average. Moreover, the proof techniques used there, such as the mass preservation property, completely break down in the presence of adversarial corruption. However, a pull-based approach like the one used in \cite{haeupler2018optimal} naturally limits the number of corrupted messages, making it a powerful paradigm to address adversarial message faults. Our algorithms additionally ensure that every node sends the same value in a given round, regardless of how many nodes pull from it in that round. This approach simplifies the problem of adversarial corruption and makes the algorithm itself simpler.

To ensure that our algorithms are suitable for realistic settings, we strive to model faults in the most general manner possible. Some fundamental questions that arise include the amount of (computational) power that a potential adversary may have or how many messages may be affected by corruptions. We need to strike a balance and allow enough power to retain a meaningful adversary, but not so much that it stifles all reasonable attempts at node communication and eliminates nontrivial results. In the following, we briefly sketch the setting, concentrating for convenience on the example of an adversary, but the model readily covers nonmalicious message corruptions.

We propose an adversary that can influence a $\beta$ fraction of the nodes in each round, which we call a \emph{$\beta$-strong adversary}, see section \ref{sec:model-algo} for a more formal definition. The adversary can corrupt any message pulled from the nodes under its influence in this round and can arbitrarily change their contents. We do not make any assumptions about the computational power of the adversary. This model covers the special case where the same $\beta n$ nodes are controlled by the adversary throughout all rounds (in other words, the case where there are $\beta n$ adversarial nodes).

With the introduction of such adversarial faults, it is not reasonable to ask the nodes to produce exact answers. Thus, the best that we can hope for is an approximate solution, which is exactly the type of problem we investigate. In the \emph{$\varepsilon$-approximate quantile problem}, the goal is to compute a value whose rank is in the interval $[(\phi - \varepsilon)n, (\phi + \varepsilon)n]$. We will start by tackling the special case $\phi=1/2$ of computing an $\varepsilon$-approximate median, and show afterwards how the problem of computing an approximation of arbitrary quantiles can be reduced to that case. We then turn to the problem of finding the (approximate) mean, which is not as robust to outliers as the median. Thus, computing the mean, even approximately, in the presence of adversarial faults is difficult, as they can sabotage any algorithm by injecting extreme outliers into the computation. The best one can hope for this \emph{$\varepsilon$-approximate mean problem} is therefore to find a value that is at most an additive $\varepsilon M$ away from the true mean, where the values are assumed to be restricted to the interval $[0, M]$.

The scenario with $\beta n$ adversarial nodes tells us that we cannot hope for an approximation factor $\varepsilon$ smaller than $\beta$. Indeed, one could simply have all the adversarial nodes reliably report an extremely large or an extremely small value, which would shift the quantiles of the ``good nodes" by $\beta$, and the same scenario can be applied to the mean approximation problem as well. Hence, all our results have $\varepsilon \ge \beta$ as an (implicit) assumption.
In some applications -- and given that a $\beta$ fraction of nodes might never output a correct answer anyway -- it might be acceptable that some other nodes also hold an incorrect answer, provided this number is small enough, especially if this accelerates the algorithm in return. To investigate the possibility of trading off speed with correctness, we analyze our algorithms from the perspective of parametrized correctness, which is new in this setting. Our gossip algorithms are said to have \emph{correctness level $\gamma$} if at most $\gamma n$ nodes are storing an incorrect value (e.g.\ a value with rank outside $[1/2-\varepsilon, 1/2+\varepsilon]$ for the $\varepsilon$-approximate median problem) upon termination.
Note that we can always assume that $\gamma \ge \frac{1}{2n}$, since $\gamma < \frac{1}{n}$ already implies that there must be no incorrect node.

\textbf{Paper Organization.} The remainder of the paper is organized as follows. We give an overview of our results in section~\ref{sec:results}, followed by a description of our model and an intuitive description of our three algorithms - complemented by their pseudocode - in section~\ref{sec:model-algo}. Section~\ref{sec:rel_work} contains a brief discussion of related work. Then, we proceed to provide a proof sketch of our main theorems in order: section~\ref{sec:approx_med} is dedicated to our median algorithm and the proof of Theorem~\ref{thm:med_find}; section~\ref{sec:shift_quant} concerns itself with our quantile shifting algorithm and the proof of Theorem~\ref{thm:shift_qnt}; section \ref{sec:approx_mean} deals with our mean algorithm. Then, section~\ref{sec:low_bound} contains the discussion of our lower bounds on the round complexity, including a proof of Theorem~\ref{thm:low_bound}. Finally, in the appendix, we first collect standard concentration results and prove some technical inequalities which are used throughout the paper, then give a complete proof of all our main results.

\section{Results}\label{sec:results}

In the remainder of the paper, we say that an event $\mathcal{E}$ occurs \emph{with high probability (w.h.p.)} if $\pr{\mathcal{E}} = 1-n^{-\Omega(1)}$. Our first main result gives a fast and robust algorithm for estimating the median. The pseudocode can be found in Algorithm~\ref{alg:3tourn} and is complemented by the runtime and correctness guarantees provided in our first theorem.

\begin{restatable}{theorem}{medfind}
     \label{thm:med_find}
     For any $\varepsilon(n), \beta(n), \gamma(n) \in (0, 1)$ satisfying $\beta \leq\frac{\varepsilon}{14}$ and $\varepsilon = \Omega \left( \frac{1} {n^{0.0019}} \right)$, there exists a gossip algorithm using messages of size $O(\log n)$ that solves the $\varepsilon$-approximate median problem in the presence of a $\beta$-strong adversary with correctness level $\gamma$ with high probability in $O \left(\log \frac{1}{\varepsilon} + \log \log \frac{1}{\gamma}  + \frac{\log (1 / \gamma)}{\log (1 / \beta)} \right)$ rounds.
\end{restatable}

Next, we investigate the problem of approximate quantiles, where the objective is to compute an $\varepsilon$-approximate $\phi$-quantile, i.e. the $\ceil{\phi}$-th smallest value. Notice that it suffices to first transform the values of the nodes such that an $\varepsilon$-approximate $\phi$-quantile in the original setting is the same as an $\varepsilon$-approximate median in this new setting. Our next main result allows for adversarially robust quantile shifting, which along with our previous result yields an adversarially robust algorithm to approximate quantiles. The pseudocode of the quantile shifting algorithm is provided in Algorithm~\ref{alg:2tourn}, the corresponding runtime guarantees are given in the theorem below.

\begin{restatable}{theorem}{shiftqnt}
    \label{thm:shift_qnt}
    For any $\varepsilon(n), \beta(n) \in (0, 1)$ satisfying $\varepsilon \leq \frac{1}{6}$, $\beta \leq \frac{\varepsilon^{2.5}}{16}$, and $\varepsilon = \Omega \big(\frac{(\log n)^{1/5}}{n^{1/5}} \big)$, there exists a gossip algorithm using messages of size $O(\log n)$ that reduces any $\varepsilon$-approximate quantile problem into an $\varepsilon$-approximate median problem in the presence of a $\beta$-strong adversary with high probability in $O \left( \log \frac{1}{\varepsilon} \right)$ rounds.
\end{restatable}

Next, we investigate the approximate mean problem with $M$-bounded values. Here, every node is initially given a value in the range $[0, M]$ (where $M$ is polynomially bounded in $n$), and the problem is to find the mean of the values up to an additive error of $\varepsilon M$. In this setting, we provide the following adversarially robust result for computing approximate means. The pseudocode for our algorithm is given in Algorithm~\ref{alg:pullavg}, its runtime and correctness guarantees are stated in the subsequent theorem.

\begin{restatable}{theorem}{mean}
\label{thm:mean}
    For any $\varepsilon(n), \beta(n), \gamma(n) \in (0, 1)$ satisfying $\beta \leq \left( \frac{\varepsilon}{100} \right)^{2.5}$ and $\varepsilon = \Omega \big( \frac{(\log n)^{6/5}}{n^{1/5}}\big)$, there exists a gossip algorithm using messages of size $O(\log n)$ that solves the $\varepsilon$-approximate mean problem with $M$-bounded values in the presence of a $\beta$-strong adversary with correctness level $\gamma$ with high probability in $O \left(\log \frac{1}{\varepsilon} + \log \frac{1}{\gamma + \beta} + \frac{\log (1 / \gamma)}{\log (1 / \beta)} \right)$ rounds.
\end{restatable}

We note that the approximation here is weaker as it depends on the promised bound on $M$. The reason for this is simple: the adversary can always claim to have a value that is very large, and this would inevitably skew the computed average to a very large value. Thus, the best one can hope for is to give an error bound as a function of a bound on the possible communicated values.

Note that in the special case where all nodes have values in $\{ 0, 1\}$, the above yields an approximate \emph{counting algorithm} with an additive error of at most $\varepsilon n$. 

Finally, we prove the following theorem, which can be used to lower-bound the round complexity of gossip algorithms in the presence of adversarial nodes.

\begin{restatable}{theorem}{lowbound}
\label{thm:low_bound}
    Let $\beta(n), \gamma(n) \in (0, 1)$. Then, in the presence of a $\beta$-strong adversary, for any gossip algorithm that runs for less than $\frac{\log (1/2\gamma)}{\log(1/\beta)}$ rounds,  with probability at least $1/10$ more than a $\gamma$ fraction of the nodes have only received corrupted messages (if any).
\end{restatable}

We also extend the $\Omega \big( \log \log n + \log\frac{1}{\varepsilon}\big)$ lower bound given in \cite{haeupler2018optimal} to an $\Omega \big( \log \log \frac{1}{\gamma} + \log \frac{1}{\varepsilon}\big)$ lower bound in our setting, see Proposition \ref{prop:lower-bound} in section \ref{sec:low_bound}. These results prove that our median and quantile algorithms have optimal round complexity and the mean algorithm has almost optimal round complexity. We note that our theorems require varying upper bounds of $\beta$ in terms of $\varepsilon$. As mentioned in the introduction, it is not possible to hope for $\varepsilon$-approximation with $\beta > \varepsilon$, as in particular all adversarial nodes might always output values that are completely out of the range of the correct values. Our analysis here is likely not tight and we did not optimize these bounds since our algorithms already cover a substantial range of values. The same holds for the upper bound on $\varepsilon$ in Theorem~\ref{thm:shift_qnt}. For lower bounds of $\varepsilon$, note that any $\varepsilon < \frac{1}{2n}$ would yield an exact algorithm, which is out of reach due to the presence of the adversary. Hence a lower bound which is inversely polynomial in $n$ follows naturally, where again we did not optimize for the precise exponent, and simply stated what was required for our analysis in order to prove concentration of the expected behavior of the algorithms. However, the key reason why we do not view this restriction as major is a different one: For many applications, approximation factors that are a small constant fraction are completely sufficient. Moreover, this is also the parameter range in which our algorithms are extremely fast. More concretely, when both $\varepsilon$ and $\gamma$ are constant, our algorithms are guaranteed to terminate in constant time, no matter how large $n$ is.

\section{Model Definition and Algorithm Overview}\label{sec:model-algo}

We consider a network of $n$ nodes that are all connected to each other (one can think of it as a complete graph on $n$ vertices). For all the problems that we study, each node is initially assigned a value that is polynomially bounded in $n$. Communication occurs in synchronous rounds, during each of which every node contacts some node chosen uniformly at random (including itself) and then either pushes or pulls a message - of size $O(\log n)$ bits - from the chosen node. As mentioned before, all our algorithms use the pull operation only, and the message sent out by a pulled node is the same, no matter which (and how many) nodes are pulling from it.

In every round, the messages sent out by at most $\beta n$ of the nodes can become corrupted. More precisely, we assume that these corruptions can be adversarial in the following sense: At the beginning of each round, the adversary can choose at most $\beta n$ nodes so that any message pulled from or pushed to these nodes can be entirely and arbitrarily modified by the adversary, that is, the message can be changed to whatever the adversary wants. Crucially, the adversary does not know beforehand the random choice of partners of the nodes in this round. We do not assume any additional restrictions on the adversary. For example, the adversary may have unbounded computational power and know the entire history of the algorithm at any given moment of time. 

Before describing our new adversarially-robust algorithms, we briefly describe the state-of-the-art in the non-adversarial setting. The approximate median find algorithm given by \cite{haeupler2018optimal} uses a pull-based tournament-style protocol. Here, every node repeatedly updates its value to the median of three randomly chosen node values. The effect is that after $O\big( \log \log n + \log\frac{1}{\varepsilon}\big)$ rounds, every node has communicated directly or indirectly with $\log n / \varepsilon^2$ other nodes. Notice that the problem becomes straightforward with large message sizes, as every node can collect all the information about every node it has directly or indirectly heard from, and $\log n / \varepsilon^2$ many uniformly randomly chosen node values are sufficient to prove the concentration of the median value for this node. As the authors show, we do not need to store all such information, and we can make do with just maintaining the median. Their algorithm then concludes with a majority vote over a constant number of nodes to allow for a union bound.

Our algorithms retain the simplicity of these previous algorithms with minor changes while being robust to adversarial corruptions. Some of the crucial changes include giving a more nuanced termination condition which takes into account $\beta$ and $\gamma$, and deriving tailored bounds on expressions involving binomial coefficients to bound the effect of the adversaries on each round. We also introduce an extended second phase -- now collecting values from $O\big( \frac{\log(1/\gamma)}{ \log(1/\beta)} \big)$ many nodes -- to ensure that at most a $\gamma$ fraction of the nodes are storing an incorrect value upon termination, as desired. 

More specifically, for the approximate median computation, as outlined above, we first let the nodes conduct a 3-tournament, that is, in the first phase, each node pulls three node values uniformly at random and updates its own value to the median of the three. In the second, additional, phase, the algorithm then uniformly samples a number of nodes and outputs their median. The pseudocode for our approximate median algorithm is given in Algorithm~\ref{alg:3tourn}.

\begin{algorithm}[H]
    \label{alg:3tourn}
    \caption{$3$-TOURNAMENT($v$)}
    \DontPrintSemicolon
    \SetAlgoLined
    $\delta \gets \left(\frac{30\log n}{n}\right)^{1/3}$ \;
    $\gamma' \gets \max(\delta, \beta, \min(\frac{\gamma}{4}, \frac{\varepsilon}{14}))$ \;
    $t \gets \ceil{\log_{(157/156)} (\frac{1}{3(\varepsilon - \beta)})} + \ceil{\log_2(\log_{9/8}(\frac{1}{\gamma'}))} + 2$ \;
    \tcp{Phase 1}
    \For{$i = 1$ to $t$}{
        Select $3$ nodes $u_1(v), u_2(v), u_3(v)$ uniformly at random \;
        $x_v \gets \mathrm{median}(x_{u_1(v)}, x_{u_2(v)}, x_{u_3(v)})$ \;
    }
    \tcp{Phase 2}
    \If{$\gamma' > \frac{\gamma}{4}$}{
        Sample $K = \ceil{8 \cdot \log_{20 \gamma'} \left( \frac{\gamma}{4} \right)}$ nodes uniformly at random and set $x_v$ to the median value of these $K$ nodes.
    }
    Output $x_v$
\end{algorithm}

To compute approximate quantiles $\phi$, \cite{haeupler2018optimal} gives another tournament-style algorithm which reduces the problem to that of finding an approximate median. Our algorithm is based on their ingeniously simple algorithm. Its basic principle can be described in one sentence: In each round, each node pulls two node values uniformly at random and updates its value to be the minimum (if $\phi < 1/2$) or maximum (if $\phi > 1/2$) of the two. Our algorithm here requires a more refined tracking of its own progress and knowing when to terminate, taking into account the magnitude of the potential corruption, which we measure in terms of $\beta$. This also greatly affects the required analysis. The pseudocode for our quantile shifting algorithm is given in Algorithm~\ref{alg:2tourn} (for the case $\phi < 1/2$).

\begin{algorithm}[H]
    \label{alg:2tourn}
    \caption{$2$-TOURNAMENT($v$)}
    \DontPrintSemicolon
    \SetAlgoLined
    $h'_0 \gets 1 - (\phi + \varepsilon)$ \;
    $i \gets 0$ \;
    $T \gets \frac{1}{2} - \frac{21\varepsilon}{16}$ \;
    \While{$h'_i > T$}{
        $h'_{i + 1} \gets \left(h'_{i} - \beta \right)^2$ \;
        $\delta \gets \min \left(1, \frac{h'_i - T}{h'_i - h'_{i + 1}} \right)$ \;
        With probability $\delta$ \textbf{do} \;
        \hspace{1em} Select $2$ nodes $u_1(v), u_2(v)$ uniformly at random \;
        \hspace{1em} $x_v \gets \min(x_{u_1(v)}, x_{u_2(v)})$ \;
        Otherwise \textbf{do} \;
        \hspace{1em} Select a node $u_1(v)$ uniformly at random \;
        \hspace{1em} $x_v \gets x_{u_1(v)}$ \;
        $i \gets i + 1$ \;
    }
\end{algorithm}

To compute the approximate mean, \cite{kempe2003gossip} gives the elegant push-sum algorithm, which works roughly as follows: In each round, every node pushes half of its value to another node which is chosen uniformly at random. After $O\big(\log n + \log \frac{1}{\varepsilon}\big)$ rounds, the algorithm converges to an approximate solution. The authors use a crucial property of this push-based protocol, mass preservation, to show that the sample variance of the values decreases exponentially, while the sample mean remains more or less unchanged.

As discussed before, the mass preservation property breaks down in the presence of adversarial corruption. Additionally, push-based protocols are not ideal as adversarially corrupted nodes gain a significant amount of influence by pushing faulty values to an arbitrary number of nodes (unless we restrict their power). Moreover, pull-based approaches generally work better than push-based ones in the presence of adversarial faults, as they inherently restrict the influence that adversarial corruptions can have. We adopt a novel approach and design a tournament-style algorithm, but use, however, a similar proof technique as in \cite{kempe2003gossip}, despite the absence of mass preservation. The algorithm itself is very simple - each node repeatedly updates its value to the average of two node values chosen uniformly at random. Finally, we also end with a majority vote, as in the approximate median algorithm, to ensure that at most a $\gamma$ fraction of the nodes output a wrong answer. The pseudocode of our adversarially robust mean approximation algorithm can be found in Algorithm~\ref{alg:pullavg}.

\begin{algorithm}[H]
    \label{alg:pullavg}
    \caption{Pull-Avg($v$)}
    \DontPrintSemicolon
    \SetAlgoLined
    $\delta \gets \frac{(\log n)^{1/2}}{n^{1/2}} $ \;
    $\eta \gets \max(\beta, \delta, \min(\gamma^5, \left(\frac{\varepsilon}{100}\right)^{2.5}))$ \;
    $T \gets \ceil{\log_{9/5} \frac{1}{\eta}}$  \;
    \tcp{Phase 1}
    \For{$t = 1$ to $T$}{
        Select $2$ nodes $u_1(v), u_2(v)$ uniformly at random \;
        If $x_{u_1(v)}, x_{u_1(v)} \notin [0, M]$, set them to the closest value in $[0, M]$ \;
        $x_v \gets \frac{x_{u_1(v)} + x_{u_2(v)}}{2}$ \;
    } 
    \tcp{Phase 2}
    \If{$\eta > \min(\gamma^5, \left(\frac{\varepsilon}{100}\right)^{2.5}))$}{
        Sample $K = \max \left(100, \ceil{40\log_{32\beta} \left( \frac{\gamma}{2} \right)} \right)$ nodes uniformly at random and set $x_v$ to the median value of these $K$ nodes.
    }
    Output $x_v$
\end{algorithm}

We highlight the trade-off between the tournament-style aggregation in the first phase and the sequential collection of values for the majority vote in the second phase. Though the former is more efficient, it also gives the adversary more (exponentially more in fact) influence over the values.

There are a variety of directions for follow-up work. In terms of other aggregation tasks, beyond computation of median, other quantiles, and mean, the main open problem which may robustly be solved in our setting is the computation of the \emph{mode} (since there is no hope for computing the minimum or maximum value in the presence of adversaries). 
We also think it would be interesting to investigate the model in the asynchronous setting or relax the assumption that communications occur on an underlying complete graph.

\section{Related Work}\label{sec:rel_work}

The Gossip Protocol is also known as the Epidemic Protocol, as these algorithms were first developed to mimic the spread of epidemics \cite{10.1145/41840.41841}. One of the first gossip algorithms studied was rumor spreading or randomized broadcast \cite{892324, 652705be-95a4-3886-b436-08eb167091e1}, where in \cite{892324} the authors give an algorithm for rumor spreading in $O(\log n)$ rounds and $O(n \log \log n)$ messages. 

In \cite{kempe2003gossip}, the authors study a gossip protocol to compute aggregates such as sums and counts in $O(\log n)$ rounds. Furthermore, they also develop algorithms for random sampling and quantile computation (also known as randomized selection), the latter of which can be computed in $O(\log^2 n)$ rounds. The work in \cite{haeupler2018optimal} improves that and gives an algorithm to compute exact quantiles in $O(\log n)$ rounds and approximate quantiles in $O\big(\log \log n + \log\frac{1}{\varepsilon} \big)$ rounds. The problem of computing quantiles has also been studied in the centralized setting \cite{BLUM1973448} and the distributed setting \cite{10.1145/1248377.1248401}. Other problems such as computing the mode \cite{10.1145/2933057.2933097} (also known as plurality consensus), have also been studied in the gossip setting. 

In both \cite{10.1007/978-3-642-14162-1_10} and \cite{892324}, the authors investigate a gossip-based rumor spreading in the presence of adaptive failures, but the adversary here only has the power to crash certain messages and cannot send altered messages (which is much more powerful). To the best of our knowledge, this is the first time a gossip model has been investigated for general aggregation problems in the presence of adversarial nodes. We note that for the special case of median computation, restricting the adversarial influence to $\beta \le \frac{1}{n^{1/2}}$, the authors of~\cite{doerr2011stabilizing} derive a runtime bound of $O(\log n \log \log n)$ for the computation of an approximate median from their approach to the stabilizing consensus problem. In this special case, at the cost of an additional factor of $O(\log \log n)$ compared to our runtime, they get an approximation factor of $O\big(\frac{(\log n)^{1/2}}{n^{1/2}}\big)$, thus yielding a tighter approximation than provided by our Theorem~\ref{thm:med_find}, which stipulates that $\varepsilon=\Omega(\frac{1}{n^{0.0019}})$. We note that we did not try to optimize the exponent in this lower bound on $\varepsilon$. We note that Byzantine-robust gossip algorithms have been widely studied in the context of decentralized machine learning ~\cite{he2023byzantinerobustdecentralizedlearningclippedgossip}, ~\cite{gaucher2025byzantinerobustgossipinsightsdual}, ~\cite{gaucher2025unifiedbreakdownanalysisbyzantine}.

In the deterministic setting, the approximate median computation in the presence of byzantine nodes has been studied~\cite{stolz2015byzantine}. The algorithm presented there was shown to produce an approximation factor of at most $\frac{\beta}{2}$ and runs in time $O(\beta n)$, as long as $\beta<\frac{1}{3}$. The latter requirement on the fraction of adversarial nodes was proven to be tight. This was later generalized to approximate the $k$ -th largest values in the presence of Byzantine nodes~\cite{melnyk2018byzantine}. 

\section{Approximate Median}\label{sec:approx_med}

In this section, we sketch the ideas behind our median approximation algorithm and the proof of its correctness and runtime guarantees stated in Theorem \ref{thm:med_find}, whose proof can be found in Appendix \ref{sec:approx_med_appendix}.

\medfind*

The algorithm, which draws on ideas from \cite{haeupler2018optimal}, is described in Algorithm \ref{alg:3tourn}, and proceeds in two phases. In the first phase, every iteration consists of $3$ rounds. More concretely, in every iteration, each node pulls the value of three random nodes and updates its value to the median of these three values. In the second phase, every node obtains $K$ many independent samples and outputs the median of these $K$ values.  

We consider the sets of nodes whose quantiles lie in $\left[ 0, \frac{1}{2} - \varepsilon\right), \left[ \frac{1}{2} - \varepsilon, \frac{1}{2} + \varepsilon \right]$ and $\left(\frac{1}{2} + \varepsilon, 1 \right]$ at the end of iteration $i$ in line $7$ of Algorithm \ref{alg:3tourn}, and call these sets $L_i,M_i,H_i$ respectively. Additionally, we also define $l_i$ to be such that $l_0 \coloneqq \frac{1}{2} - \varepsilon$ and $l_{i + 1} \coloneqq \overline{B}_{3, 1}(l_i + \beta) = 3(l_i + \beta)^2 - 2(l_i + \beta)^3$.\footnote{$\overline{B}_{3, 1}(l_i + \beta)$ is defined as $1-B_{3, 1}(l_i + \beta)$ with $B_{n,k}(p)$ the Cumulative Distribution Function of the binomial distribution with parameters $n,k,p$.  For details, see the Appendix.} These values capture the expected values of $\frac{|L_i|}{n}$ as we will show later. We remark that $H_i$ and $L_i$ behave symmetrically, hence without loss of generality, we focus on bounding $|L_i|$ in the following proof. In the end, we use this symmetry to infer matching bounds for $|H_i|$.

\textbf{Proof Sketch.} We first show that, at the end of the first phase, that is, after $t$ rounds (line $8$ of the pseudocode of Algorithm~\ref{alg:3tourn}), $l_t$ is small enough (smaller than $\gamma'$ as defined in line 2 of Algorithm \ref{alg:3tourn}). The proof is a straightforward induction based on the definition of $l_i$, but we utilize tailored bounds on the expressions (involving binomial coefficients) that arise from the recursive definition of $l_i$. 

Next, we show that $\ex{\frac{|L_{i + 1}|}{n}\ \mid\ L_i} \leq \overline{B}_{3, 1}\left(\frac{|L_i|}{n} + \beta \right)$. This follows from an analysis of the main loop of the algorithm. Then, using Chernoff's bounds, we show that $\frac{|L_{i + 1}|}{n}$ is concentrated around its expectation. 

Notice that $l_i$ is defined recursively in such a way as to mirror the evolution of $\frac{|L_{i}|}{n}$ in a single step. In fact, we make this connection more explicit by showing that for all $0 \leq i \leq t$, $|L_i|/n$ is approximately equal to $l_i$ in expectation and with high probability. Thus, we use the bound on $l_t$ to conclude that after the first phase, at most a $\gamma'$ fraction of the nodes have a wrong answer. 

Finally, the second phase amplifies the number of correct nodes in each step and precisely runs the number of steps needed to reduce this fraction of bad nodes from $\gamma'$ to $\gamma$.

\section{Shifting Quantiles}\label{sec:shift_quant}

In this section, we give some intuition underlying our quantile shifting algorithm and the proof of its correctness and runtime guarantees given in Theorem \ref{thm:shift_qnt}, whose proof can be found in Appendix \ref{sec:shift_quant_appendix}.

\shiftqnt*

Note that we can assume $\phi \leq 1/2$ without loss of generality, as the algorithm and analysis are exactly symmetric for the other case.

The pseudocode for our algorithm is described in Algorithm \ref{alg:2tourn}. We consider the sets of nodes whose quantiles lie in $\left[ 0, \phi - \varepsilon\right), \left[ \phi - \varepsilon, \phi + \varepsilon \right]$ and $\left(\phi + \varepsilon, 1 \right]$ at the end of iteration $i$ of the while loop and call these sets $L_i,M_i,H_i$ respectively. We want to show that $|L_i|$ and $|H_i|$ decrease rapidly as the algorithm progresses.

The values $h_i'$ in Algorithm \ref{alg:2tourn} determine when the algorithm terminates, and will also be helpful for the analysis. We define another quantity recursively as 
\[h_0 \coloneqq h_0', \qquad h_{i+1} \coloneqq (h_i+\beta)^2,\]
which we will use for the analysis as well. As will be shown later, the values $h_i, h_i'$ capture the expected values of $\frac{|H_i|}{n}$.

\textbf{Proof Sketch.} Let $t$ be the total number of iterations of the while-loop executed by Algorithm \ref{alg:2tourn}. We first show that $t$ is bounded by $\mathcal{O} \left( \log \frac{1}{\varepsilon}\right)$ as required. Notice that $h_i, h_i'$ are defined so that they start out the same ($h_0 = h_0'$) and drift apart very slowly. We make this explicit by showing that after $t$ steps, they are within about a $(1 + \varepsilon)$ factor of each other. Next, we show that $\left(\frac{|H_i|}{n} - \beta \right)^2 \leq \ex{\frac{|H_{i + 1}|}{n}\ \mid\ H_i} \leq \left(\frac{|H_i|}{n} + \beta \right)^2$. This follows from an analysis of the main loop of the algorithm. 

As in the median find algorithm, notice that the recursive definition of $h_i', h_i$ mirrors the evolution of $\frac{|H_i|}{n}$. In fact, we make this explicit by giving an upper and lower bound of $\frac{|H_i|}{n}$ in terms of $h_i, h_i'$, respectively. This is the most technically involved part of the proof, utilizing tedious but straightforward bounds that we tailored for this analysis. Consequently, from the termination condition $h_t' \leq T \coloneqq \frac{1}{2} - \frac{21 \varepsilon}{16}$, we deduce that $\frac{|H_t|}{n}$ is around $\frac{1}{2} - \varepsilon$.

The final part of the proof involves showing that $|M_i|$ does not change significantly during the entire algorithm and stays around its initial value of $2 \varepsilon$. This follows from an analysis of the expected value of $|M_{i + 1}|$ given $M_i$ and then using Chernoff's bounds. This would imply that $\frac{|L_t|}{n}$ is also around $\frac{1}{2} - \varepsilon$, which concludes our proof.

\section{Approximate Mean}\label{sec:approx_mean}

In this section, we give a proof sketch for the runtime and correctness guarantees of our adversarially robust approximate mean algorithm stated in Theorem~\ref{thm:mean}. The complete proof can be found in Appendix \ref{sec:approx_mean_appendix}.

\mean*

The pseudocode is described in Algorithm \ref{alg:pullavg}. In the first phase of the algorithm, in each round, every node $u$ samples two nodes and updates its value to their arithmetic mean. Samples of node value that exceed the acceptable value interval $[0,M]$ are cut to the interval limits. In the second phase, the algorithm globally samples a set of nodes uniformly at random and outputs the median value of these nodes. Let $x_v(t)$ denote the value $x_v$ stored in the node $v$ at the end of iteration $t$. Additionally, let us denote by $V$ the set of all nodes. 

To help analyze the progress of the algorithm, we define the potential function.
\[ \Phi(t) \coloneqq \sum_{u \neq v \in V} (x_u(t) - x_v(t))^2,\]
where the sum covers all subsets of $V$ of size $2$ (that is, we do not double-count the pairs).
Notice that this potential captures the variation in the values stored in the nodes. As the algorithm progresses, we expect the potential to decrease to a small value as the values stored in the nodes converge to the average.
We also define the quantity $\psi(t) \coloneqq \sum_{u \in V} x_u(t)$ as the sum of the node values. Notice that $\frac{\psi(0)}{n}$ is the true average that we want to approximate. 

\textbf{Proof Sketch.} Intuitively, the $2$-tournament algorithm must decrease the variation of the node values in each step as each node updates its value to the average of two other node values. Through a more precise analysis, we show that $\Phi(t)$ approximately halves in expectation in each step. Thus, by using Azuma's inequality, we can obtain concentration on this expectation and show that at the end of the $T$ steps it has been reduced to less than about a $\beta$ fraction of its original value.

On the other hand, in each step, the influence of the adversary is limited to a $\beta$ fraction of the nodes. This suffices to limit the change in $\psi(t)$ in a single step to about $\beta n M$ in expectation. Using Azuma's inequality again, we can bound the total change in $\psi(t)$ after $T$ steps by about $\varepsilon n M$.

If the average of the values does not change much but the variance decreases substantially, we can conclude that most of the values must be close to the average. This is exactly what we make next explicit, showing that the number of nodes that can have a wrong answer at the end of the first phase is bounded by a polynomial in $\beta$.

Finally, as in the median find algorithm, the second phase amplifies the number of correct nodes in each step and precisely runs the number of steps needed to reduce this fraction of bad nodes to $\gamma$.

\section{Lower Bounds on Round Complexity}\label{sec:low_bound}

In this section, we prove Theorem~\ref{thm:low_bound}, which yields a general tool for calculating the lower bounds for the round complexity of gossip algorithms in the presence of adversaries. In essence, the theorem quantifies a lower bound on the number of rounds needed for all but a $\gamma$ fraction of the nodes to produce a correct answer, in the presence of $\beta n$ adversaries. Note that a similar lower bound was proven for the non-adversarial setting in~\cite{haeupler2018optimal}, but its proof does not extend readily to our case because our parametrized notion of correctness allows a $\gamma$ fraction of nodes to be incorrect in addition to the $ \beta n$ adversarial nodes. However, what remains true is that even with the slack provided by $\gamma$, a correct algorithm cannot terminate in fewer than $\Omega\big(\frac{\log(1/\gamma)}{\log(1/\beta)}\big)$ rounds, since up to this point, more than a $\gamma$ fraction of nodes will have communicated only with bad nodes. Therefore, they will have received no accurate information whatsoever and could not possibly output the correct result. This is formalized in Theorem \ref{thm:low_bound}.

\lowbound*

\begin{proof}
    First, suppose that $2\gamma > \beta$. Then, notice that $\frac{\log(1/2\gamma)}{\log(1/\beta)} < 1$ and hence the claim holds trivially. Thus, we can now assume that $2 \gamma \leq \beta$.

    Consider a single good node $v$. In one round, the probability of communicating with an adversarially affected node is $\beta$. Thus, after $t < \frac{\log (1/2\gamma)}{\log(1/\beta)}$ rounds, the probability that it has communicated only with adversarial nodes is at least $\left(\beta \right)^t > 2\gamma$. Thus, in expectation after $t$ rounds, there is at least a $2 \gamma$ fraction of such nodes. 

    Now, notice that the communication partner chosen by an arbitrary node in an arbitrary round is completely independent of what is chosen by any other node in any round. Thus, by the Chernoff bounds, the probability that there is at least a $\gamma$ fraction of the nodes that have only communicated with bad nodes during the first $t < \frac{\log (1/2\gamma)}{\log(1/\beta)}$ rounds is at least
    \[1 - \exp \left(- \frac{2\gamma n}{2^2 \cdot 2} \right) \geq 1 - e^{-1/8} > \frac{1}{10},\]
    where we assumed without loss of generality that $\gamma \geq \frac{1}{2n}$.
\end{proof}

Notice that we immediately get a $\Omega \big( \frac{\log(1/\gamma)}{\log(1/\beta)} \big)$ lower bound for the tasks that we study in this paper. This is because any node that has only communicated with adversarial nodes cannot reliably output a correct answer.

Similarly, one can adapt the proof of \cite[Theorem 1.3]{haeupler2018optimal} to show a $\Omega \left( \log \log \frac{1}{\gamma} + \log \frac{1}{\varepsilon} \right)$ lower bound for these problems. We formalize this statement in Proposition \ref{prop:lower-bound} and give its proof in Appendix \ref{sec:low_bound_appendix}.

\begin{restatable}{proposition}{lowerboundprop}\label{prop:lower-bound}
    For any $\varepsilon \in (\frac{10 \log n}{n}, \frac{1}{6})$ and any $\gamma \in (\frac{40e \log n}{n}, \frac{1}{2})$, any gossip algorithm (even with unlimited message size) that uses less than $\log_2 \log_{4e} \frac{1}{\gamma}+\log_4 \frac{1}{6\varepsilon}$ rounds fails to solve either of the $\varepsilon$-approximate median or (bounded) mean problem with probability at least $1/3$. 
\end{restatable}

Together, Theorem \ref{thm:low_bound} and Proposition \ref{prop:lower-bound} show that our algorithm for the approximate median is optimal. Moreover, our algorithm for the approximate mean is also almost optimal with only an extra $O (\log \frac{1}{\gamma + \beta})$ factor and is optimal whenever $\varepsilon$ is polynomially smaller than $\gamma$ or $\beta$.

\bibliographystyle{alpha}
\bibliography{refs}

\appendix

\section{Tools and Bounds}

In this section, we collect useful tools and bounds that we use for our analysis throughout the paper. We begin with a statement of the standard Chernoff bounds.

\begin{lemma}(Chernoff Bounds)
    \label{lem:chernoff}
    Suppose that $X_1, X_2, \ldots, X_n$ are independent indicator random variables, and let $X = \sum_{i = 1}^n X_i$ be their sum. Then, for $\delta > 0$,
    \[ \pr{X \geq (1 + \delta)\ex{X}} \leq \left( \frac{e^{\delta}}{(1 + \delta)^{(1 + \delta)}}\right)^{\ex{X}},\]
    \[ \pr{X \leq (1 - \delta)\ex{X}} \leq \left( \frac{e^{\delta}}{(1 - \delta)^{(1 - \delta)}}\right)^{\ex{X}}.\]
    The following simpler statements hold for $0 < \delta < 1$
    \[ \pr{X \geq (1 + \delta)\ex{X}} \leq e^{-\delta^2 \ex{X}/ 3}, \]
    \[ \pr{X \leq (1 - \delta)\ex{X}} \leq e^{-\delta^2 \ex{X}/ 2}. \]
\end{lemma}

We also state the following extension of the Chernoff bounds for larger deviations, which we get by using \cite[Lemma A.2]{haeupler2018optimal}.

\begin{lemma}
    \label{lem:chernplus}
    Suppose that $X_1, X_2, \ldots, X_n$ are independent indicator random variables, and let $X = \sum_{i = 1}^n X_i$ be their sum. Then, for some $M \geq \ex{X}$ and $0 < \delta < 1$,
    \[ \pr{X > \ex{X} + \delta M} \leq \left( \frac{e^{\delta}}{(1 + \delta)^{(1 + \delta)}}\right)^{M} \leq e^{-\delta^2 M / 3}.\]
\end{lemma}
\begin{proof}
    Lemma \ref{lem:chernoff} and \cite[Lemma A.2]{haeupler2018optimal} immediately imply the desired claim.
\end{proof}

Next, we state Azuma's inequality.

\begin{lemma}(Azuma's Inequality)
    \label{lem:azuma}
    Let $\Omega$ be the product of $m$ probability spaces $\Omega_1, \ldots, \Omega_m$, and let $Z: \Omega\to\R$ be a random variable. Suppose there is $c>0$ such that for all $\omega, \omega' \in \Omega$ that differ in exactly one component $\Omega_i$, we have $|Z(\omega)-Z(\omega')| \le c$. Then for all $t > 0$
    \[ \pr{\abs{Z - \ex{Z}} \geq t} \leq 2 e^{-\frac{t^2}{2mc^2}}. \]
    
\end{lemma}

We will now prove some useful properties of binomial distributions that we will use later in the analysis of the algorithms. Let $\mathcal{B}(n, p)$ denote the binomial distribution for the parameters $n \geq 1$ and $0 \leq p \leq 1$. Let us denote by $b_{n, k}(p)$ the Probability Distribution Function (PDF) of $\mathcal{B}(n, p)$, i.e.\ $b_{n, k}(p) \coloneqq \pr{\mathcal{B}(n, p) = k}$, and by $B_{n, k}(p)$ the Cumulative Distribution Function (CDF) of $\mathcal{B}(n, p)$, i.e.\ $B_{n, k}(p)\coloneqq\pr{\mathcal{B}(n, p) \leq k}$. We have for $0 \leq k \leq n$,
\[ b_{n, k}(p) = \binom{n}{k} p^{k} (1 - p)^{n - k},\]
\[ B_{n, k}(p) = \sum_{0 \leq l \leq k} b_{n, l}(p).\]
For convenience we also define $\overline{B}_{n, k}(p) \coloneqq 1 - B_{n, k}(p)$. We remark that $\overline{B}_{n, k}(p) = \pr{\mathcal{B}(n, p) > k}$ is monotonically increasing in $p$.

We are interested in the case when $p = \frac{1}{2} - \alpha$ for some small $\alpha > 0$. Our objective is to prove an upper bound on $\overline{B}_{n, k}(p)$ which we will be useful in later analysis.

\begin{lemma}
    \label{lem:scale}
    For any $0 \leq p \leq 1 \leq x$ such that $px \leq 1$, we have
    $b_{n, k}(xp) \leq x^n b_{n, k}(p)$ and $\overline{B}_{n, k}(xp) \leq x^n \overline{B}_{n, k}(p)$.
\end{lemma}
\begin{proof}
    We have
    \[ b_{n, k}(xp) = \binom{n}{k} (xp)^{k} (1 - xp)^{n - k} \leq \binom{n}{k} (xp)^{k} (x(1 - p))^{n - k} = x^n b_{n, k}(p).\]
    Thus,
    \[ \overline{B}_{n, k}(xp) = \sum_{k < l \leq n} b_{n, l}(xp) \leq \sum_{k < l \leq n} x^n b_{n, l}(p) = x^n \overline{B}_{n, k}(p).\]
\end{proof}

\begin{lemma}
    \label{lem:alpha}
    Suppose that $0 \leq \alpha \leq \frac{1}{3}$. Then for any integer $n \geq 2$ we have $(1 - 2 \alpha)^n \leq 1 - \frac{8}{3} \cdot \alpha$.
\end{lemma}
\begin{proof}
    The claim is clearly true for $\alpha = 0$. For $\alpha > 0$, note that
    \[ (1 - 2\alpha)^n = 1 - f(n, \alpha) \cdot \alpha\]
    where
    \[ f(n, \alpha) \coloneqq \frac{1 - (1 - 2\alpha)^n}{\alpha} = 2 \cdot \sum_{i = 0}^{n - 1} (1 - 2\alpha)^i \]
    is a function that is monotonically increasing in $n$, and monotonically decreasing in $\alpha$.
    Thus
    \[ (1 - 2\alpha)^n \leq 1 - f \left(2, \frac{1}{3} \right) \cdot \alpha = 1 - \frac{8}{3} \cdot \alpha.\]
\end{proof}

Next, we prove another upper bound for exponentials.

\begin{lemma}
    \label{lem:beta}
    Suppose that $0 \leq x \leq M$. Then we have
    \[ e^{\alpha x} \leq 1 + \frac{e^{M \alpha} - 1}{M} x .\]
\end{lemma}
\begin{proof}
    The claim is clearly true for $x = 0$. For $x > 0$, note that 
    \[ e^{\alpha x} = 1 + f(x) x\]
    where
    \[ f(x) \coloneqq \frac{e^{\alpha x} - 1}{x} \]
    is monotonically increasing in $x$ in its domain. Thus, we have
    \[ e^{\alpha x} \leq 1 + f(M)x \leq 1 + \frac{e^{M\alpha} - 1}{M} x .\]
\end{proof}

\begin{lemma}
    \label{lem:h_func}
    Consider the functions $f(x) \coloneqq \left( \frac{x}{1 + \beta} \right)^2$ and $g(x) \coloneqq \left( \frac{x + \beta}{1 + \beta} \right)^2$ for some $0 < \beta < 1$. Then, for all $\beta^2 \leq x \leq 1$, we have $f(x) \leq x$ and $g(x) \leq x$.
\end{lemma}
\begin{proof}
    First, consider the function $F(x) \coloneqq f(x) - x$. We have $F(0) = f(0) - 0 = 0$, and $F((1 + \beta)^2) = f((1 + \beta)^2) - (1 + \beta)^2 = 0$. Thus, since $F$ is a quadratic in $x$ with a positive leading coefficient, we must have $F(x) \leq 0$ for all $0 \leq x \leq (1 + \beta)^2$. Consequently, $f(x) \leq x$ for all $\beta^2 \leq x \leq 1$.

    Next, consider the function $G(x) \coloneqq g(x) - x$. We have $G(1) = g(1) - 1 = 0$, and $G(\beta^2) = g(\beta^2) - \beta^2 = 0$. Thus, since $G$ is a quadratic in $x$ with a positive leading coefficient, we must have $G(x) \leq 0$ for all $\beta^2 \leq x \leq 1$. Consequently, $g(x) \leq x$ for all $\beta^2 \leq x \leq 1$.
\end{proof}

\begin{theorem}
    \label{thm:cdf}
    Suppose that $k \geq 0$ is an integer and $p = \frac{1}{2} - \alpha$ for some $0 \leq \alpha \leq \frac{1}{3}$. Then, we have 
    \[\overline{B}_{(2k + 1), k}(p) \leq \frac{1}{2} - \frac{13}{12} \cdot \alpha .\]
\end{theorem}
\begin{proof}
    For an arbitrary integer $l \in [k + 1, 2k + 1]$, consider the expression $p^l \left( 1 - p\right)^{2k + 1 - l}$. By definition $p = \frac{1}{2} - \alpha$, which reduces the expression to
    \[ \left( \frac{1}{2} - \alpha \right)^l \left( \frac{1}{2} + \alpha \right)^{2k + 1 - l} = \frac{1}{2^{2k + 1}} \cdot \left( 1 - 4\alpha^2 \right)^{2k + 1 - l} \left( 1 - 2\alpha \right)^{2l - 2k - 1}.\]
    Using the fact that $1 - 4\alpha^2 \leq 1$ we can simplify the expression as follows
    \[ p^l \left( 1 - p\right)^{2k + 1 - l} \leq \frac{1}{2^{2k + 1}} \cdot \left( 1 - 2\alpha \right)^{2l - 2k - 1}.\]
    For $l > k + 1$, notice that $2l - 2k - 1 > 1$, and hence we can use Lemma \ref{lem:alpha} to further bound the expression as
    \[ p^l \left( 1 - p\right)^{2k + 1 - l} \leq \frac{1}{2^{2k + 1}} \cdot \left(1 - \frac{8}{3} \cdot \alpha \right) .\]
    Putting it together in the definition of $\overline{B}_{(2k + 1), k}(p)$, we obtain
    \begin{align*}
        \overline{B}_{(2k + 1), k}(p) &\coloneqq \sum_{l = k + 1}^{2k + 1} \binom{2k + 1}{l} p^l \left(1 - p\right)^{2k + 1 - l} \\
        &\leq \frac{1}{2^{2k + 1}} \cdot \left [ \binom{2k + 1}{k + 1}(1 - 2 \alpha) + \sum_{l = k + 2}^{2k + 1} \binom{2k + 1}{l} \left( 1 - \frac{8}{3} \cdot \alpha \right) \right]
    \end{align*}
    We now use a well known identity of the binomial coefficients, $\sum_{l = k + 1}^{2k + 1}\binom{2k + 1}{l} = 2^{2k}$, to deduce that $\sum_{l = k + 2}^{2k + 1}\binom{2k + 1}{l} = 2^{2k} - \binom{2k + 1}{k + 1}$. Substituting in this, we obtain
    \begin{align*}
        \overline{B}_{(2k + 1), k}(p) &\leq \frac{1}{2^{2k + 1}} \cdot \left [ \binom{2k + 1}{k + 1}(1 - 2 \alpha) + \left(2^{2k} - \binom{2k + 1}{k + 1} \right) \left( 1 - \frac{8}{3} \cdot \alpha \right) \right] \\
        &= \left(\frac{1}{2} - \frac{4}{3} \cdot \alpha \right) + \frac{1}{2^{2k + 1}} \cdot \left [ \binom{2k + 1}{k + 1}(1 - 2 \alpha) - \binom{2k + 1}{k + 1} \left( 1 - \frac{8}{3} \cdot \alpha \right) \right] \\
        &= \left(\frac{1}{2} - \frac{4}{3} \cdot \alpha \right) + \frac{1}{2^{2k + 1}} \cdot \binom{2k + 1}{k + 1} \cdot \frac{2}{3} \cdot \alpha .
    \end{align*}
    For all $k \geq 0$, consider the ratio $r_k \coloneqq \frac{1}{2^{2k + 1}}\binom{2k + 1}{k + 1}$. Since we have
    $\frac{r_{k + 1}}{r_k} = \frac{(2k + 3)(2k + 2)}{4(k + 2)(k + 1)} = 1 - \frac{1}{2(k + 2)} < 1$, the ratio $r_k$ must be monotonically decreasing in $k$. Thus, for all $k \geq 1$, we have $r_k \leq r_1 = \frac{3}{8}$. Thus, we have
    \[ \overline{B}_{(2k + 1), k}(p) \leq \left( \frac{1}{2} - \frac{4}{3}  \cdot \alpha \right) + \frac{3}{8} \cdot \frac{2}{3} \cdot \alpha = \frac{1}{2} - \frac{13}{12} \cdot \alpha.\]
\end{proof}

\section{Approximate Median}\label{sec:approx_med_appendix}

In this section we prove the runtime and correctness guarantees for our adversarially-robust median algorithm, as stated in Theorem \ref{thm:med_find}.

\medfind*

Remember the following recursive definition:
\[
l_0 \coloneqq \frac{1}{2} - \varepsilon, \qquad
l_{i + 1} \coloneqq \overline{B}_{3, 1}(l_i + \beta).
\]
We begin with a lemma that bounds $l_t$.

\begin{lemma}\label{lem:l_t<beta}
    For $t = \ceil{\log_{(157/156)} (\frac{1}{3(\varepsilon - \beta)})} + \ceil{\log_2(\log_{9/8}(\frac{1}{\gamma'}))} + 2$, we have $l_t \leq \gamma'$.
\end{lemma}
\begin{proof}
    First, suppose that for some $i \geq 0$, we have $l_{i} + \beta = \frac{1}{2} - \alpha$ where $\alpha$ is such that $0 < \left(\frac{157}{156} \right)^{i}(\varepsilon - \beta) \leq \alpha \leq \frac{1}{3}$. Then, we will show that $l_{i + 1} + \beta = \frac{1}{2} - \alpha'$ where $\alpha'$ is such that $0 < \left(\frac{157}{156} \right)^{i + 1}(\varepsilon - \beta) \leq \alpha'$.
    Notice that under this assumption, we have
    \[l_{i+1} = \overline{B}_{3, 1}(l_{i} + \beta) = \overline{B}_{3, 1} \left(\frac{1}{2} - \alpha \right) .\]
    Now, since $\alpha$ satisfies the conditions required to use Theorem \ref{thm:cdf}, we have
    \[ l_{i+1} \leq \frac{1}{2} - \frac{13}{12} \cdot \alpha \leq \frac{1}{2} - \frac{13}{12} \cdot \left(\frac{157}{156} \right)^{i}(\varepsilon - \beta)\]
    where we have used the lower bound on $\alpha$ in the last step. Next, notice that $\frac{13}{12} = \frac{157}{156} + \frac{1}{13}$. Additionally, since we have assumed that $\beta \leq \varepsilon / 14$ in our setting, we have $\varepsilon - \beta \geq 13 \beta$. Combining these two, we have
    \begin{align*}
        l_{i+1} &\leq \frac{1}{2} - \left(\frac{157}{156} \right)^{i + 1}(\varepsilon - \beta) - \frac{1}{13} \cdot \left(\frac{157}{156} \right)^{i} \cdot (\varepsilon - \beta) \\
        &\leq \frac{1}{2} - \left(\frac{157}{156} \right)^{i + 1}(\varepsilon - \beta) - \frac{1}{13} \cdot (\varepsilon - \beta) \\
        &\leq \frac{1}{2} - \left(\frac{157}{156} \right)^{i + 1}(\varepsilon - \beta) - \beta
    \end{align*}
    Notice that for $i = 0$, we have $l_0 + \beta = \frac{1}{2} - (\varepsilon - \beta)$. Thus, using iterated applications of the previous claim, we can deduce that the first iteration $t_1$ where $l_{t_1} + \beta \leq \frac{1}{2} - \frac{1}{3}$ satisfies $t_1 \leq \ceil{\log_{(157/156)} (\frac{1}{3(\varepsilon - \beta)})}$. Therefore after $t_1 = \ceil{\log_{(157/156)} (\frac{1}{3(\varepsilon - \beta)})}$ iterations, we have $l_{t_1} \leq \frac{1}{2} - \frac{1}{3} - \beta = \frac{1}{6} - \beta$.
    
    Now, consider some $i \geq t_1$ and suppose that $l_i + \beta \leq \frac{1}{6}$ and $l_i \geq \beta$. Then
    \[ l_{i + 1} \leq 3 (l_i + \beta)^2 \leq 3 \cdot \frac{1}{6} \cdot \left(l_i + l_i\right) \leq l_i.\]
    This means that the sequence $l_i$ is non-increasing in $i$, which is what we desire.

    Now, suppose that $l_i \leq 3 \beta$ for some $i \geq 0$. Then, by our assumption that $\beta \leq \frac{\varepsilon}{14} \leq \frac{1}{28}$, we have
    \[ l_{i + 1} \leq 3(l_i + \beta)^2 \leq 3(4\beta)^2 = (48\beta) \beta \leq 2\beta .\]
        Furthermore, 
    \[ l_{i + 2} \leq 3(l_{i + 1} + \beta)^2 \leq 3(3\beta)^2 = (27\beta) \beta \leq \beta .\]
    
    Hence, we can deduce that for $i \geq t_1$, $l_i$ monotonically decreases until it reaches $\beta$ and then it stays below $\beta$. Additionally, it only takes two steps to drop from $3\beta$ to below $\beta$.

    Next, for any $j \geq 0$, suppose that $l_{t_1 + j} \geq 3\beta$. Then, we have
    \[ l_{t_1 + j + 1} \leq 3 \left(l_{t_1 + j} + \beta \right)^2 \leq 3 \left(l_{t_1 + j} + \frac{l_{t_1 + j}}{3}\right)^2 \leq \frac{16}{3} (l_{t_1 + j})^2.\]
    Hence, if $t_1 + t_2$ is the first iteration when $l_{t_1 + t_2} < 3\beta$, we can deduce that
    \[ l_{t_1 + t_2} \leq \frac{16}{3} (l_{t_1 + t_2 - 1})^2 \leq \ldots \leq \left(\frac{16}{3}\right)^{2^{t_2} - 1} (l_{t_1})^{2^{t_2}} \leq \left(\frac{16 l_{t_1}}{3}\right)^{2^{t_2}}.\]

    Recall that $l_{t_1} \leq \frac{1}{6}$. Thus, 
    \[ l_{t_1 + t_2} \leq \left(\frac{16 l_{t_1}}{3}\right)^{2^{t_2}} \leq \left(\frac{8}{9}\right)^{2^{t_2}}.\]
    Thus, for $t_2 = \ceil{\log_2(\log_{9/8}(\frac{1}{\gamma'}))} + 2$, we conclude that $l_{t_1 + t_2} = l_{t} \leq \gamma'$. Notice that this requires that assumption that $\gamma' \geq \beta$, which is true by definition.  

\end{proof}

In the next two lemmas, we prove a bound on the expected size of $L_{i + 1}$ given $L_i$ and show that the true value is concentrated around this expectation.

\begin{lemma}
    \label{lem:exp}
    For each iteration $0 \leq i \leq t$, we have $\ex{\frac{|L_{i + 1}|}{n} \mid L_i} \leq \overline{B}_{3, 1} \left(\frac{|L_{i}|}{n} + \beta \right)$.
\end{lemma}
\begin{proof}
    For $i \geq 0$, notice that a node $v$ is in $L_{i + 1}$ only if at least $2$ of the $3$ nodes $u_1(v), u_2(v), u_3(v)$ are either in $L_i$ or adversarially affected in round $i$. Indeed, if two or more nodes are in $M_i, H_i$, then the median of the three will be in $M_{i + 1}, H_{i + 1}$. 
    Thus,
    \[ \ex{\frac{|L_{i + 1}|}{n} \mid L_i} \leq \overline{B}_{3, 1} \left(\frac{|L_{i}|}{n} + \beta \right)\]
    where the bound follows from the fact that $\overline{B}_{n, k}(p)$ is monotonically increasing in $p$.
\end{proof}

\begin{lemma}
    \label{lem:oneround}
    W.h.p.\ for each iteration $1 \leq i \leq t$, $\frac{|L_i|}{n} \leq (1 + \delta) \max \left( \gamma', \ex{\frac{|L_{i}|}{n} \mid L_{i - 1}} \right)$.
\end{lemma}
\begin{proof}
    In each iteration $i$, the event that an arbitrary node is in $L_i$ at the end of this round is conditionally independent from that of other nodes given $L_{i-1}$. Hence by applying Lemma \ref{lem:chernplus} with $M = \max \left( \gamma' n, \ex{|L_{i}| \mid L_{i - 1}} \right) \geq \ex{|L_{i}| \mid L_{i - 1}}$, we have
    \[ \pr{ |L_i| \geq (1 + \delta)M } \leq \exp \left( -\frac{\delta^2 M}{3}\right) \]
    Now, since $\delta^2 M \geq \delta^2 \gamma' n \geq \delta^3 n = 30\log n$ (recall that $\gamma' = \max(\delta, \beta, \min(\frac{\gamma}{4}, \frac{\varepsilon}{14}))$ and $\delta = (30\log n/n)^{1/3}$), we get that $\pr{ |L_i| \geq (1 + \delta)M } \leq n^{-10}$.
    Finally, using a union bound over all $t$ iterations, we get the required claim.
\end{proof}

Next, we bound the cumulative error in the value of $|L_i| / n$.

\begin{lemma}\label{lem:oneround-concentration}
    W.h.p.\ for each iteration $0 \leq i \leq t$, $\frac{|L_i|}{n} \leq \max \left( (1 + \delta)\gamma', (1 + \delta)^{\frac{3^i - 1}{2}} \cdot l_i\right)$.
\end{lemma}
\begin{proof}
    First, we notice that $t \le \frac{\log (1/\varepsilon)}{\log (157/156)} + O(\log \log n)$ and recall that $\varepsilon = \Omega \left( \frac{1}{n^{0.0019}} \right)$ by assumption. Combining these two, we have 
    \[3^t = O \left( \varepsilon^{-\frac{\log 3}{\log (157/156)}} \cdot \mathrm{polylog}(n)\right) = O(n^{0.33}),\] 
    and thus for all $i\le t$,
    \[ (1 + \delta)^{\frac{3^i - 1}{2}} \leq \exp \left( \delta \cdot 3^t \right) = \exp (o(1)) \leq 2.\]

    Now, we will prove the required claim by induction on $i$. For $i = 0$, clearly $\frac{|L_0|}{n} \leq l_0$. From Lemma \ref{lem:oneround}, we know that $\frac{|L_{i+1}|}{n} \leq (1 + \delta) \max \left( \gamma', \ex{\frac{|L_{i+1}|}{n} \mid L_{i}} \right)$. 
    
    Suppose that $\frac{|L_{i}|}{n} \leq (1 + \delta)\gamma' \leq 2 \gamma'$. Then, $\frac{|L_{i}|}{n} + \beta \leq 3\gamma'$. Additionally, notice that $\beta \leq \frac{\varepsilon}{14} \leq \frac{1}{27}$ and $\delta = o(1) \leq \frac{1}{27}$. Thus, we must have $\gamma' \coloneqq \max(\delta, \beta, \min(\frac{\gamma}{4}, \frac{\varepsilon}{14})) \leq \frac{1}{27}$. These two facts along with Lemma \ref{lem:exp} gives us the following
    \[\ex{\frac{|L_{i + 1}|}{n} \mid L_i} \leq \overline{B}_{3, 1} \left(\frac{|L_{i}|}{n} + \beta \right) \leq \overline{B}_{3, 1} \left(3\gamma' \right) \leq 3(3 \gamma')^2 \leq \gamma'.\]
    
    Now, suppose that $\frac{|L_{i}|}{n} > (1 + \delta)\gamma'$. By induction hypothesis, it must be the case that $\frac{|L_{i}|}{n} \leq (1 + \delta)^{\frac{3^i - 1}{2}} \cdot l_i$. We have
    \[ \ex{\frac{|L_{i+1}|}{n} \mid L_{i}} \leq \overline{B}_{3, 1} \left(\frac{|L_{i}|}{n} + \beta \right) \leq \overline{B}_{3, 1} \left( (1 + \delta)^{\frac{3^i - 1}{2}} \cdot l_i + \beta \right) \leq \overline{B}_{3, 1} \left( (1 + \delta)^{\frac{3^i - 1}{2}} \cdot (l_i + \beta) \right).\]
    Notice that $(1 + \delta)^{\frac{3^i - 1}{2}} \cdot (l_i + \beta) \leq 2(l_i + \beta) \leq 1$, which justifies its usage in the final step. Now, we use Lemma \ref{lem:scale} to obtain
    \[ \ex{\frac{|L_{i+1}|}{n} \mid L_{i}} \leq (1 + \delta)^{3 \cdot \frac{3^i - 1}{2}} \cdot  \overline{B}_{3, 1} \left( l_i + \beta \right) = (1 + \delta)^{\frac{3^{i + 1} - 3}{2}} \cdot l_{i + 1}.\]
    Thus, we have 
    \[\frac{|L_{i + 1}|}{n} \leq (1 + \delta) \ex{\frac{|L_{i}|}{n} \mid L_{i - 1}} \leq (1 + \delta)^{\frac{3^{i + 1} - 1}{2}} \cdot l_{i + 1} .\]
\end{proof}

\begin{corollary}
    \label{cor:bound}
    W.h.p.\ $\frac{|L_t|}{n} \leq 2 \cdot \max \left( \gamma', l_t\right)$.
\end{corollary}
\begin{proof}
    Recall from the proof of Lemma \ref{lem:oneround-concentration} that $(1 + \delta)^{\frac{3^t - 1}{2}} \leq 2$. Using this along with the result of Lemma \ref{lem:oneround-concentration} for the iteration $t$, we obtain the required claim.
\end{proof}

Notice that since $L_i$ and $H_i$ are both completely symmetric, the above corollary also applies to $H_i$.

Now, we can analyze the number of nodes that output the wrong answer at the end of the second phase.

\begin{lemma}\label{lem:final-median-approx}
    W.h.p.\ all but a $\gamma$ fraction of the nodes outputs a quantile in $\left[ \frac{1}{2} - \varepsilon, \frac{1}{2} + \varepsilon \right]$
\end{lemma}
\begin{proof}
    By Corollary \ref{cor:bound} and Lemma \ref{lem:l_t<beta}, we know that after $t$ iterations, we have $\frac{|L_i|}{n} \leq 2\gamma'$ and $\frac{|H_i|}{n} \leq 2\gamma'$. Now, suppose that $\gamma' \leq \frac{\gamma}{4}$. Then, we are already done. 
    
    Otherwise, we must have that $\gamma < 4 \max(\beta, \delta)$. This is because if $\gamma' > \gamma / 4$, then we must have $\max(\beta, \delta, \min\{\gamma/4,\varepsilon/14\}) > \gamma / 4$, which implies that $\max(\beta, \delta, \gamma/4) > \gamma / 4$, and hence $\max(\beta, \delta) > \gamma / 4$. A fixed node $v$ outputs a quantile in the required range if at least $K/2$ of the nodes are in $M_t$ (line~9 of Algorithm \ref{alg:3tourn}). This means that the probability that this does not happen is at most 
    \[
    \binom{K}{K/2} \left( \frac{|L_i|}{n} + \frac{|H_i|}{n} + \beta \right)^{K/2} \leq 2^K (5\gamma')^{K/2} = (20\gamma')^{K/2} \leq \left(\frac{\gamma}{4} \right)^4,
    \]
    where the last inequality holds since $K = \ceil{8 \cdot \log_{20 \gamma'} \left( \frac{\gamma}{4} \right)}$. Notice that $K > 0$ because $20 \gamma' \leq \frac{20}{27} < 1$. Thus, the expected number of nodes that do not output a correct approximate median is at most $\frac{\gamma^4 n}{256}$. 
    
    First, suppose that $\gamma = \Omega \left( \frac{\log n}{n}\right)$. Notice that the expected number of nodes is $\frac{\gamma^4 n}{256} \leq \frac{\gamma n}{2}$. Now, using Chernoff bounds, the probability that the number of nodes that output the wrong answer is more than $\gamma n$ is upper bounded by $\exp(-\frac{1}{3} \cdot  (\frac{\gamma n}{2})) = \exp(- \Omega(\log n)) = n^{-\Omega(1)}$.

    Else, we have $\gamma = o \left( \frac{\log n}{n}\right)$. Now, we can apply Markov's to say that the probability that the expected number of nodes that output the wrong answer is more than $\gamma n$ is upper bounded by $\frac{\gamma^4 n}{256 \cdot \gamma n} = o \left( \frac{\log^3 n}{n^3} \right) = n^{-\Omega(1)}$.
\end{proof}

This concludes the proof of the correctness of the algorithm. Finally, we analyze the runtime of the algorithm.

\begin{theorem}
    \label{thm:runtime}
    The  algorithm terminates in $O \left(\log \frac{1}{\varepsilon} + \log \log \frac{1}{\gamma}  + \frac{\log \gamma}{\log \beta} \right)$ rounds.
\end{theorem}
\begin{proof}
    Notice that the total number of rounds is $O(t + K) = O \left( \log \frac{1}{\varepsilon} + \log \log \frac{1}{\gamma'} + \frac{\log (1/\gamma)}{ \log (1/\gamma')} \right)$.

    Recall that $\gamma' \coloneqq \max(\delta, \beta, \min(\frac{\gamma}{4}, \frac{\varepsilon}{14})) \geq \min(\frac{\gamma}{4}, \frac{\varepsilon}{14})$. It follows that 
    \[O \left( \log \log \frac{1}{\gamma'} \right) = O\left( \log \log \frac{1}{\gamma} + \log \log \frac{1}{\varepsilon}\right).\]
    Hence the runtime is $O \left( \log \frac{1}{\varepsilon} + \log \log \frac{1}{\gamma} + \frac{\log (1/\gamma)}{ \log (1/\gamma')} \right)$.

    Next, suppose that $\gamma' = \gamma / 4$. Then $O \left( \frac{\log (1/\gamma)}{ \log (1/\gamma')} \right) = O(1)$. Otherwise if $\gamma' = \delta$, then $O \left( \frac{\log (1/\gamma)}{ \log (1/\gamma')} \right) = O \left( \frac{\log (1/\gamma)}{ \log n} \right) = O(1)$ (since we can always assume that $\gamma = \Omega(1/n)$). Notice that if $\gamma' = \varepsilon / 14$, then $\gamma / 4 \geq \varepsilon / 14$ and hence $O \left( \frac{\log (1/\gamma)}{ \log (1/\gamma')} \right) = O \left( \frac{\log (1/\gamma)}{ \log (1/\varepsilon)} \right) = O \left( 1 \right)$. Finally, if $\gamma' = \beta$, we have $O \left( \frac{\log (1/\gamma)}{ \log (1/\gamma')} \right) = O \left( \frac{\log (1/\gamma)}{ \log (1/\beta)} \right)$. Thus, $O \left( \frac{\log (1/\gamma)}{ \log (1/\gamma')} \right) = O \left(1 +  \frac{\log \gamma}{ \log \beta} \right)$ always holds.

    Consequently, the runtime can be written as $O \left(\log \frac{1}{\varepsilon} + \log \log \frac{1}{\gamma}  + \frac{\log \gamma}{\log \beta} \right)$.
\end{proof}

\section{Shifting Quantiles}\label{sec:shift_quant_appendix}

In this section, we prove the theoretical runtime guarantees of Theorem \ref{thm:shift_qnt} for our robust quantile-shifting algorithm.
\shiftqnt*

Note that the assumptions $\beta \leq \frac{\varepsilon^{2.5}}{16}$ and $\varepsilon \leq \frac{1}{6}$ together imply that $\beta \le \frac{\varepsilon \cdot \varepsilon^{3/2}}{16} \le \frac{\varepsilon}{16} \cdot \frac{1}{6\sqrt{6}}= \frac{\varepsilon}{96 \sqrt{6}} \leq \frac{\varepsilon}{200}$.
We begin with a deterministic runtime bound for Algorithm \ref{alg:2tourn}.

\begin{lemma}
    \label{lem:iter}
    Let $t$ denote the number of iterations performed by Algorithm \ref{alg:2tourn} until termination. We have $t \leq (1.5) \log_2 (\frac{1}{\varepsilon})$.
\end{lemma}
\begin{proof}
    The algorithm ends when $h'_t \leq \frac{1}{2} - \frac{21\varepsilon}{16}$. First, for some arbitrary $i$, suppose that $h'_{i} \leq 1 - \left( \frac{7}{4}\right)^{i} \varepsilon$ and $\left( \frac{7}{4}\right)^{i} \varepsilon \leq \frac{1}{4}$. We will show that $h'_{i + 1} \leq 1 - \left( \frac{7}{4}\right)^{i + 1} \varepsilon$. We have
    \[ h'_{i + 1} \coloneqq  (h'_i - \beta)^2 \leq (h'_i)^2 \leq \left( 1 - \left( \frac{7}{4}\right)^i\varepsilon\right)^2 .\]

    Next, notice that for all $0 \leq x \leq \frac{1}{4}$, we have $(1 - x)^2 = 1 - 2x + x^2 \leq 1-2x + \frac{1}{4}x = 1 - \left( \frac{7}{4}\right)x$. Using this, we obtain
    \[ h'_{i + 1} \leq 1 - \left( \frac{7}{4}\right) \cdot \left( \frac{7}{4}\right)^{i}\varepsilon = 1 - \left( \frac{7}{4}\right)^{i + 1}\varepsilon.\]

    Initially, we have $h_0 \leq 1 - (\phi + \varepsilon) \leq 1 - \varepsilon$. Thus, by iterated applications of the previous claim, after at most $\log_{7/4} \left(\frac{1}{4\varepsilon} \right)$ many iterations, we must have that $h'_i \leq \frac{3}{4}$.
    Furthermore, after $3$ additional iterations, $h'_i$ drops below $\frac{1}{4}$. 
    
    Now, suppose that $t$ is the first iteration when $h'_{t} \leq \frac{1}{4}$. Thus, after 
    \[t \leq \log_{7/4} \left(\frac{1}{4\varepsilon} \right) + 3 \leq (1.5) \log_{2} \left(\frac{1}{4\varepsilon} \right) + 3 = (1.5) \log_{2} \left(\frac{1}{\varepsilon} \right)\] 
    iterations, we have $h'_{t} \leq \frac{1}{4} \leq \frac{1}{2} - \frac{21\varepsilon}{16}$ for any $\varepsilon \leq \frac{1}{6}$.
\end{proof}

We mention the following claim that will be useful.
\begin{claim}
    \label{claim:iter}
    $h'_{t - 1} \geq \frac{1}{2} - \frac{21\varepsilon}{16} \geq \frac{9}{32}$.
\end{claim}

Note that lines $8, 9$ in Algorithm \ref{alg:2tourn} are executed with probability $1$ during the first $t - 1$ iterations. Indeed, $\delta$ is set to $1$ if and only if $h'_i-h'_{i+1} \le h'_i-T$ or equivalently $T \le h'_{i+1}$, which is true up to $i=t-1$ by definition of $t$. Recall the following recursive definitions:
\begin{align*}
    h_0' \coloneqq 1-(\phi+\varepsilon), \qquad h_{i+1}' \coloneqq (h_i'-\beta)^2, \\
    h_0 \coloneqq 1-(\phi+\varepsilon), \qquad h_{i+1} \coloneqq (h_i+\beta)^2.
\end{align*}
We now prove that $h_t, h'_t$ are close to each other, which we will later use to show that $H_t$ is concentrated around its expectation.

\begin{lemma}
    \label{lem:hbound}
    After the penultimate iteration, we have $h'_{t - 1} \leq h_{t - 1} \leq h'_{t - 1}\left( 1 + \frac{3\varepsilon}{4} \right)$. Similarly, at the end of the algorithm, we have $h'_{t} \leq h_{t} \leq h'_{t}\left( 1 + \frac{3\varepsilon}{4}\right)$.
\end{lemma}
\begin{proof}   
The lower bounds for $h_{t - 1}, h_{t}$ are straightforward to prove using induction on $i$. To do this, we will prove the stronger claim that $h'_i \leq h_i$ for all $i$. Clearly, $h'_0 = h_0$. For the induction step, we have 
\[h'_{i + 1} \coloneqq \left( h'_{i} - \beta \right)^2 \leq \left( h'_{i} + \beta \right)^2 \leq \left( h_{i} + \beta \right)^2 = h_{i + 1}.\]
The upper bounds on $h_{t - 1}$ is trivial for $t \leq 1$ since $h_0 = h'_0$, so we will assume that $t \geq 2$ in the following. 

For any $1 \leq i \leq t$, we have
\[ h_{i} \coloneqq \left( h_{i - 1} + \beta\right)^2 = \left( h_{i - 1}\right)^2 \cdot \left( 1 + \frac{\beta}{h_{i - 1}}\right)^2 \leq \left( h_{i - 1} \right)^2 \cdot e^{2\beta / h_{i - 1}}.\] 
Similarly, notice that
\[ h'_{i} \coloneqq \left( h'_{i - 1} - \beta\right)^2 = \left( h'_{i - 1}\right)^2 \cdot \left( 1 - \frac{\beta}{h'_{i - 1}}\right)^2 \geq \left( h'_{i - 1} \right)^2 \cdot e^{-3\beta / h'_{i - 1}} \]
where in the final step, we have used the fact that $1 - x \geq e^{-3x}$ for all $x \leq 1/2$. This inequality indeed applies here because $h'_{i - 1} \geq 2 \beta$ always holds, as otherwise we have $h'_{t - 1} < 2\beta \leq \varepsilon / 100 < 1/2 - 21\varepsilon / 16 \eqqcolon T$, which is a contradiction. 

For the upper bound on $h_{t - 1}$, notice that by Claim \ref{claim:iter}, $h_{t - 1} \geq h'_{t - 1} \geq 9/32$. Now, we will show that for all $0 \leq i \leq t - 2$, we have $h'_i > 10/19$. Suppose that by contradiction $h'_i \leq 10/19$ for some such $i$. This must mean that $h'_{i + 1} \coloneqq \left( h'_{i} - \beta \right)^2 \leq \left( \frac{10}{19}\right)^2 < \frac{9}{32}$, which is a contradiction. Thus, we have $h_i \geq h'_i \geq 10/19$ for all $0 \leq i \leq t - 2$. Using this and applying the bounds on $h_{t-1}$ and $h'_{t - 1}$ iteratively, we get 
\begin{align*} 
h_{t - 1} \leq \left( h_{t - 2}\right)^2 \cdot e^{(1.9)(2\beta)} \leq &\ldots \leq  h_0^{2^{t - 1}} \cdot e^{(1.9)(2^{t} - 2) \beta}, \\
h'_{t - 1} \geq \left( h'_{t - 2}\right)^2 \cdot e^{-(1.9)(3\beta)} \geq &\ldots \geq  (h'_0)^{2^{t - 1}} \cdot e^{-(2.85)(2^{t} - 2) \beta}.
\end{align*}

Rearranging the second inequality and using the fact that $h_0 = h'_0$ gives us $h_0^{2^{t - 1}} \leq h'_{t - 1} \cdot e^{(2.85)(2^{t} - 2) \beta}$. Hence, we have
\[h_{t - 1} \leq  h_0^{2^{t - 1}} \cdot e^{(1.9)(2^{t} - 2) \beta} \leq h'_{t - 1} \cdot e^{(4.75)\cdot (2^{t}-2) \beta} \leq h'_{t - 1} \cdot e^{(4.75)\cdot \varepsilon^{-1.5} \beta}  \leq h'_{t - 1} \cdot e^{3\varepsilon / 10},\]
where we have used Lemma \ref{lem:iter} for the second inequality and our assumption that $\beta \leq \frac{\varepsilon^{2.5}}{16}$ for the last inequality. Using Lemma \ref{lem:beta} with $x = \varepsilon, \alpha = 3/10, M = 1/6$ then proves the upper bound on $h_{t - 1}$.

For the upper bound on $h_t$, we note that $h_{t - 1} \geq h'_{t - 1} \geq 9/32 \geq 1/4$ by Claim \ref{claim:iter}. Thus, similar to before, we have
\[ h_t \leq \left( h_{t - 1}\right)^2 \cdot e^{8\beta} \leq \left( h'_{t - 1}\cdot e^{3\varepsilon / 10} \right)^2  \cdot e^{8\beta} = (h_{t-1}')^2 \cdot e^{3\varepsilon/5} \cdot e^{8\beta} \leq h'_t \cdot e^{3\varepsilon / 5} \cdot e^{8\beta} \cdot e^{12\beta} \leq h'_t \cdot e^{7 \varepsilon / 10}\]
where we have used our assumption $\beta \leq \varepsilon / 200$ for the last inequality. Using Lemma \ref{lem:beta} with with $x = \varepsilon, \alpha = 7/10, M = 1/6$ then proves the upper bound on $h_{t}$.
\end{proof}

Next, we bound the expected value of $|H_i|$ in each iteration of the algorithm.

\begin{lemma}
    \label{lem:exp2}
    For iteration $1 \leq i \leq t - 1$, $\left( \frac{|H_i|}{n} - \beta\right)^2 \leq \ex{\frac{|H_{i + 1}|}{n} \mid H_i} \leq \left( \frac{|H_i|}{n} + \beta \right)^2$
\end{lemma}
\begin{proof}
In the iteration $i$, for a given node $v$, consider the two nodes $u_1(v), u_2(v)$ that it picks in this iteration. Recall that for $i \le t-1$, the algorithm indeed always samples two nodes, selecting them independently of each other. Note that for $v$ to be in $H_{i + 1}$, it is sufficient for $u_1(v), u_2(v)$ to both be in $H_i$ and not be adversarially affected in this round. Since each node sampled at time $i$ has a probability of $\frac{|H_i|}{n}$ of being in $H_i$ and has a probability of at most $\beta$ of being adversarially affected, the lower bound then follows. For the upper bound, notice that it is necessary for $u_1(v), u_2(v)$ to both be in either $H_i$ or be adversarially affected.
\end{proof}

Now, we analyze what happens in the final iteration of the algorithm. As we show in the next lemma, if our proxy $h'_{t-1}$ of the fraction of nodes with values in $H_{t-i}$ is correct up to a small additive error $\varepsilon''$, then in expectation the fraction $H_t$ in the next (final) iteration will be almost as close to $T$.

\begin{lemma}
    \label{lem:final}
    Suppose that $(1 - \varepsilon'')h'_{t - 1} \leq \frac{|H_{t - 1}|}{n} \leq (1 + \varepsilon'')h_{t - 1}$ for some $16 \beta \leq \varepsilon'' \leq 1/9$. Then we have $T - \frac{3}{2} \cdot \varepsilon'' \leq \ex{\frac{|H_{t}|}{n} \mid H_{t - 1}} \leq T + \frac{3}{2} \cdot\varepsilon'' + \frac{1}{2} \cdot \varepsilon$. 
\end{lemma}
\begin{proof}
    By analogous reasoning as in the proof of the lower bound in Lemma \ref{lem:exp2}, now taking into account that with probability $1-\delta$ we only sample one node, which is in $H_{t-1}$ with the same probability as above, we have
    \[ \ex{\frac{|H_{t}|}{n} \mid H_{t - 1}} \geq (1 - \delta) \cdot \left(\frac{|H_{t - 1}|}{n} - \beta \right) + \delta \cdot \left(\frac{|H_{t - 1}|}{n} - \beta \right)^2. \]
    Now, recall from Claim \ref{claim:iter} that $h'_{t - 1} \geq \frac{9}{32}$. By assumption, $\varepsilon'' \leq \frac{1}{9}$. Furthermore, by assumption, $\frac{|H_{t - 1}|}{n} \geq (1 - \varepsilon'')h'_{t - 1} \geq \frac{8}{9} \cdot \frac{9}{32} \geq \frac{1}{4}$. Thus,
    \begin{align*}
    \ex{\frac{|H_{t}|}{n} \mid H_{t - 1}} &\geq (1 - \delta) \cdot \frac{|H_{t - 1}|}{n} \cdot \left(1 - \frac{\beta n}{|H_{t - 1}|} \right) + \delta \cdot \left(\frac{|H_{t - 1}|}{n}\right)^2 \cdot \left(1 - \frac{\beta n}{|H_{t - 1}|} \right)^2 \\
    &\geq (1 - 4\beta) \left[(1 - \delta) \cdot \frac{|H_{t - 1}|}{n} + \delta \cdot \left(\frac{|H_{t - 1}|}{n}\right)^2(1 - 4\beta)\right] \\
    &\geq (1 - 4\beta) \left[(1 - \delta) \cdot (1 - \varepsilon'')h'_{t - 1} + \delta \cdot (1 - \varepsilon'')^2(h'_{t - 1})^2(1 - 4\beta)\right] \\
    &\geq (1 - 4\beta)^2 (1 - \varepsilon'')^2\left[(1 - \delta) \cdot h'_{t - 1} + \delta \cdot (h'_{t - 1})^2\right] \\
    &\geq (1 - 8\beta) (1 - \varepsilon'')^2 h_{t-1}'\left[(1 - \delta)  + \delta \cdot h'_{t - 1}\right] \\
    &\geq (1 - \varepsilon'')^3 h_{t-1}'\left[(1 - \delta) + \delta \cdot h'_{t - 1}\right],
    \end{align*}
    where in the last line we have used that $8\beta \le 16\beta \le \varepsilon''$.
    
    Now, using the facts that $\delta = \frac{h'_{t - 1} - T}{h'_{t - 1} - h'_t}$ and $h_t' \leq (h_{t-1})^2$, we have
    \begin{align*}
    \ex{\frac{|H_{t}|}{n} \mid H_{t - 1}} 
    &\geq (1-\varepsilon'')^3 h_{t-1}' \left[\frac{T-h_t'}{h_{t-1}'-h_t'} + \frac{h_{t-1}'-T}{h_{t-1}'-h_t'} \cdot h'_{t - 1}\right] \\
    &= (1-\varepsilon'')^3 \left[\frac{T-h_t'}{h_{t-1}'-h_t'} \cdot h_{t-1}' + \frac{h_{t-1}'-T}{h_{t-1}'-h_t'} \cdot (h'_{t - 1})^2\right] \\
    &\geq (1-\varepsilon'')^3 \left[\frac{T-h_t'}{h_{t-1}'-h_t'} \cdot h_{t-1}' + \frac{h_{t-1}'-T}{h_{t-1}'-h_t'} \cdot h'_{t}\right]
    = (1 - \varepsilon'')^3 \cdot T.
    \end{align*}
    Finally, we use some simple bounds. The first one is $(1 - \varepsilon'')^3 \geq 1 - 3 \varepsilon''$ , which holds for all $0\le \varepsilon'' \le 1$. The second one is $T \leq 1/2$ which holds by definition since $\varepsilon\ge 0$. We obtain
    \begin{align*}
        \ex{\frac{|H_{t}|}{n} \mid H_{t - 1}} 
        \geq (1 - 3\varepsilon'') \cdot T \geq T - \frac{3}{2} \cdot \varepsilon'',
    \end{align*}
    which concludes the proof of the lower bound.

    For the upper bound, we have
    \[ \ex{\frac{|H_{t}|}{n} \mid H_{t - 1}} \leq (1 - \delta) \cdot \left(\frac{|H_{t - 1}|}{n} + \beta \right) + \delta \cdot \left(\frac{|H_{t - 1}|}{n} + \beta \right)^2. \]
    Furthermore, using our assumption that $\frac{|H_{t - 1}|}{n} \leq (1 + \varepsilon'')h_{t - 1}$, we have 
    \begin{align*}
    \ex{\frac{|H_{t}|}{n} \mid H_{t - 1}} &\leq (1 - \delta) \cdot ((1 + \varepsilon'')h_{t - 1} + \beta) + \delta \cdot ((1 + \varepsilon'')h_{t - 1} + \beta)^2 \\
        &\leq (1 + \varepsilon'')^2 \cdot \left[ (1 - \delta) \cdot (h_{t - 1} + \beta) + \delta \cdot (h_{t - 1} + \beta)^2 \right]\\
        &= (1 + \varepsilon'')^2 \cdot \left[ (1 - \delta) \cdot (h_{t - 1} + \beta) + \delta \cdot h_t \right]\\
    \end{align*}
    where we have used the fact that $h_t \coloneqq (h_{t - 1} + \beta)^2$. 
    Now, notice that by Claim \ref{claim:iter} and Lemma \ref{lem:hbound}, we have $h_{t - 1} \geq h'_{t - 1} \ge 9/32 > 1/4$. Thus, we have 
    \[h_{t - 1} + \beta \leq h_{t - 1} + 4h_{t - 1}\beta = h_{t - 1} \cdot (1 + 4\beta).\]
    Using this we obtain the following bound
    \begin{align*}
       \ex{\frac{|H_{t}|}{n} \mid H_{t - 1}} &\leq (1 + \varepsilon'')^2 \cdot (1 + 4\beta) \cdot \left[ (1 - \delta) \cdot h_{t - 1} + \delta \cdot h_t \right] \\
       &\leq (1 + \varepsilon'')^2 \cdot (1 + 4\beta) \cdot \left[ (1 - \delta) \cdot h'_{t - 1} + \delta \cdot h'_t \right] \cdot \left( 1 + \frac{3 \varepsilon}{4} \right) \\
       &\leq (1 + \varepsilon'')^2 \cdot (1 + 4\beta) \cdot  T \cdot \left( 1 + \frac{3\varepsilon}{4} \right).
    \end{align*}
    where in the penultimate step we have used Lemma \ref{lem:hbound} and in the final step we have used the fact that $\delta = \frac{h'_{t - 1} - T}{h'_{t - 1} - h'_t}$. Now, using the assumption that $\varepsilon'' \geq 16 \beta$, we get
    \[ (1 + \varepsilon'')^2 \cdot (1 + 4\beta) \leq (1 + \varepsilon'')^2 \cdot (1 + \varepsilon'' / 4) \leq 1 + (9\varepsilon'' + 6 (\varepsilon'')^2 + (\varepsilon'')^3)/4 \leq 1 + 3 \varepsilon'', \]
    where in the final step we have used the the assumption that $\varepsilon'' \leq 1/9$. Thus, we have
    \begin{align*}
        \ex{\frac{|H_{t}|}{n} \mid H_{t - 1}} &\leq (1 + 3\varepsilon'') \cdot  T \cdot \left( 1 + \frac{3\varepsilon}{4} \right)
        = T + 3\varepsilon'' \cdot  T + (1 + 3\varepsilon'') \cdot  T \cdot \frac{3\varepsilon}{4}.
    \end{align*}
    Finally, we use the inequalities $T \leq 1/2$ and $\varepsilon'' \leq 1/9$ to obtain
    \begin{align*}
        \ex{\frac{|H_{t}|}{n} \mid H_{t - 1}} &\leq T + 3\varepsilon'' \cdot  \frac{1}{2} + \left(1 + 3 \cdot \frac{1}{9} \right) \cdot \frac{1}{2} \cdot \frac{3\varepsilon}{4}
        = T + \frac{3}{2} \cdot\varepsilon'' + \frac{1}{2} \cdot \varepsilon.
    \end{align*}
    
\end{proof}

Next, we will show that the cumulative error is not too large.

\begin{lemma}
    \label{lem:cumerr}
    Let $\varepsilon' = \frac{\varepsilon}{9 \cdot 2^{t + 1}}$, where $\varepsilon = \Omega \left(\left(\frac{ \log n}{n} \right)^{1/5} \right)$. W.h.p.\ for iteration $0 \leq i < t$, we have $\frac{1}{4} \leq (1 - \varepsilon')^{2^{i + 1} - 1}h'_{i} \leq \frac{|H_i|}{n} \leq (1 + \varepsilon')^{2^{i + 1} - 1}h_{i}$.
\end{lemma}
\begin{proof}
    First, notice that for all $0 \leq i < t$, we have
    \[ (1 - \varepsilon')^{2^{i + 1} - 1}h'_{i} \geq (1 - \varepsilon')^{2^{t} - 1}h'_{t - 1}.\]
    Now, using the fact that $(1 - \varepsilon') \geq e^{-3\varepsilon' /2 }$ (because $1 - x \geq e^{-3x/2}$ for all $x \leq 1/2$), and the fact that $h'_{t - 1} \geq \frac{9}{32}$ (Claim \ref{claim:iter}) we have
    \[ (1 - \varepsilon')^{2^{t} - 1}h'_{t - 1} \geq e^{-3\varepsilon' \cdot 2^{t - 1}} \cdot \frac{9}{32} \geq e^{-\varepsilon/ 12} \cdot \frac{9}{32} \geq e^{-1/72} \cdot \frac{9}{32} \geq \frac{1}{4}. \]
    Hence, we must have $(1 - \varepsilon')^{2^{i + 1} - 1}h'_{i} \geq 1/4$ for all $0 \leq i < t$, which proves the first inequality in the desired claim.

    We will prove the bounds on $\frac{|H_i|}{n}$ by induction on $i$. Note that by Lemma \ref{lem:iter}, we have $2^{t} \leq \varepsilon^{-1.5}$ and hence
    \begin{align}\label{eq:eps'-lower-bound}
        \varepsilon' = \frac{\varepsilon}{9 \cdot 2^{t+1}} \geq \frac{\varepsilon^{2.5}}{18} = \Omega \left(\frac{(\log n)^{1/2}}{n^{1/2}} \right).
    \end{align}
    Hence the base case holds since
    \[ (1 - \varepsilon')h'_{0} \leq (1 - 1/n)h'_{0} \leq \frac{|H_0|}{n} \leq (1 + 1/n)h_{0} \leq (1 + \varepsilon')h_{0}.\]
    For the induction step, suppose that the claim is true for some $i \geq 0$. We have, using Lemma \ref{lem:exp2} in the first step, the lower bound $\frac{|H_i|}{n}\ge\frac{1}{4}$ in the second and third steps,
    \[\ex{\frac{|H_{i + 1}|}{n} \mid H_i} \geq \left( \frac{|H_i|}{n} - \beta \right)^2 \geq \left( \frac{|H_i|}{n} \right)^2(1 - 4\beta)^2 \geq \left( \frac{1}{4} \right)^2(1 - 4\beta)^2 = \Omega(1).\]
    
    Using \eqref{eq:eps'-lower-bound} and the Chernoff bounds, we have 
    \[ \pr{\left|\frac{|H_{i + 1}|}{n}-\ex{\frac{|H_{i + 1}|}{n} \mid H_i}\right| > \frac{\varepsilon' }{2}\cdot \ex{\frac{|H_{i + 1}|}{n} \mid H_i}} \leq 2 \cdot e^{-\frac{\varepsilon'^2}{12}\ex{|H_{i + 1}| \mid H_i}} \leq e^{- \Omega(\frac{\log n}{n} \cdot n)} = n^{- \Omega(1)} .\]
    Thus, w.h.p.\ $\left|\frac{|H_{i + 1}|}{n}-\ex{\frac{|H_{i + 1}|}{n} \mid H_i}\right| \leq \frac{\varepsilon' }{2}\cdot \ex{\frac{|H_{i + 1}|}{n} \mid H_i}$.

    Now, we use Lemma \ref{lem:exp2} and the induction hypothesis to obtain w.h.p.\
    \begin{align*}
       \frac{|H_{i + 1}|}{n} &\geq (1 - \varepsilon' / 2) \left( \frac{|H_i|}{n} - \beta\right)^2 \\
       &\geq (1 - \varepsilon' / 2) \left( (1 - \varepsilon')^{2^{i + 1} - 1} h'_i - \beta\right)^2\\
       &= (1 - \varepsilon' / 2) (1 - \varepsilon')^{2^{i + 2} - 2} \left( h'_i - (1 - \varepsilon')^{-2^{i + 1} - 1} \beta\right)^2\\
       &\geq (1 - \varepsilon' / 2) (1 - \varepsilon')^{2^{i + 2} - 2} \left( h'_i - e^{3\varepsilon' 2^{i}} \beta \right)^2\\
       &\geq (1 - \varepsilon' / 2) (1 - \varepsilon')^{2^{i + 2} - 2} \left( h'_i - e^{3\varepsilon' 2^{t - 1}} \beta \right)^2\\
       &\geq (1 - \varepsilon' / 2) (1 - \varepsilon')^{2^{i + 2} - 2} \left( h'_i - e^{\varepsilon / 12} \beta \right)^2,
    \end{align*}
    where in the fourth step we have used the fact that $(1 - \varepsilon') \geq e^{-3\varepsilon' /2 }$, and in the last two steps we have used the fact that $\varepsilon' 2^i \leq \varepsilon' 2^{t-1} = \varepsilon/36$. Furthermore,
    \begin{align*}
        \frac{|H_{i + 1}|}{n} &\geq (1 - \varepsilon' / 2) (1 - \varepsilon')^{2^{i + 2} - 2} \left( h'_i - e^{\varepsilon / 12} \beta \right)^2 \\
        &\geq (1 - \varepsilon' / 2) (1 - \varepsilon')^{2^{i + 2} - 2} \left( h'_i - (1 + \varepsilon / 8) \beta \right)^2\\
        &= (1 - \varepsilon' / 2) (1 - \varepsilon')^{2^{i + 2} - 2}  \left( h'_i -  \beta \right)^2 \left( 1 - \frac{\varepsilon \beta}{8(h'_i -  \beta)} \right)^2\\
        &\geq (1 - \varepsilon' / 2) (1 - \varepsilon')^{2^{i + 2} - 2} \left( h'_i -  \beta \right)^2 \left( 1 - \varepsilon \beta\right)^2\\
        &\geq (1 - \varepsilon' / 2) (1 - \varepsilon')^{2^{i + 2} - 2}  \left( h'_i -  \beta \right)^2 \left( 1 - 2\varepsilon \beta \right)\\
        &\geq (1 - \varepsilon' / 2) (1 - \varepsilon')^{2^{i + 2} - 2}  \left( h'_i -  \beta \right)^2 \left( 1 - \beta / 3 \right)\\
        &\geq (1 - \varepsilon' / 2) (1 - \varepsilon')^{2^{i + 2} - 2}  \left( h'_i -  \beta \right)^2 \left( 1 - \varepsilon' / 4 \right)\\
        &\geq (1 - \varepsilon')^{2^{i + 2} - 1}  \left( h'_i -  \beta \right)^2\\
        &= (1 - \varepsilon')^{2^{i + 2} - 1} h'_{i + 1},
    \end{align*}
    where in the seventh step we used \eqref{eq:eps'-lower-bound} and the assumption that $\beta \le \frac{\varepsilon^{2.5}}{16}$ to get $\varepsilon' \geq \frac{\varepsilon^{2.5}}{18} \geq \frac{8 \beta}{9}$.
    Similarly, we have
    \begin{align*}
       \frac{|H_{i + 1}|}{n}  &\leq (1 + \varepsilon' / 2) \ex{\frac{|H_{i + 1}|}{n} \mid H_i} \\
       &\leq (1 + \varepsilon' / 2) \left( \frac{|H_i|}{n} + \beta\right)^2 \\
       &\leq (1 + \varepsilon' / 2) \left( (1 + \varepsilon')^{2^{i + 1} - 1} h_i + \beta\right)^2\\
       &\leq (1 + \varepsilon' / 2) (1 + \varepsilon')^{2^{i + 2} - 2} \left( h_i +  \beta\right)^2\\
        &\leq (1 + \varepsilon')^{2^{i + 2} - 1}  \left( h_i +  \beta \right)^2\\
        &= (1 + \varepsilon')^{2^{i + 2} - 1} h_{i + 1}.
    \end{align*}
\end{proof}

Finally, we can bound the cumulative deviation at the end of the algorithm.

\begin{lemma}
    \label{lem:Hbound}
    Let $\varepsilon = \Omega \left(\left(\frac{\log n}{n} \right)^{1/5} \right)$. At the end of the algorithm, w.h.p.\ $\frac{1}{2} - \frac{15\varepsilon}{8} \leq \frac{|H_t|}{n} \leq \frac{1}{2} - \frac{\varepsilon}{8}$.
\end{lemma}
\begin{proof}
    Let $\varepsilon' = \frac{\varepsilon}{9\cdot 2^{t+1}}$. By Lemma \ref{lem:cumerr}, we have w.h.p.
    \begin{align}\label{eq:H_{t-1}-bounds}
        (1 - \varepsilon')^{2^t - 1} h'_{t - 1} \leq \frac{|H_{t-1}|}{n} \leq (1 + \varepsilon')^{2^t - 1} h_{t - 1}.
    \end{align} 
    Let $\varepsilon'' \coloneqq 3 \cdot 2^{t}\varepsilon' = \frac{\varepsilon}{6}$. Notice that, by the choice of $\varepsilon''$, $16 \beta \leq \frac{\varepsilon}{6} = \varepsilon''$, and $\varepsilon'' = \frac{\varepsilon}{6} \leq \frac{1}{9}$. Since $2^{t} \varepsilon' = \frac{\varepsilon}{9} \leq \frac{1}{2}$, we have $(1 + \varepsilon')^{2^t - 1} \leq 1 + (3 \cdot (2^{t} - 1))\varepsilon' \leq 1 + \varepsilon''$ and $(1 - \varepsilon')^{2^t - 1} \geq 1 - (3 \cdot (2^{t} - 1))\varepsilon' \geq 1 - \varepsilon''$. Combining this with \eqref{eq:H_{t-1}-bounds} and Lemma \ref{lem:final} yields
    \[ T - 9 \cdot 2^{t - 1} \varepsilon' \leq \ex{\frac{|H_t|}{n} \mid H_{t - 1}} \leq T + 9 \cdot 2^{t - 1}\varepsilon' + \frac{\varepsilon}{2}.\]
    In particular $\ex{\frac{|H_t|}{n} \mid H_{t - 1}} = \Theta(1)$ (i.e.\ $\ex{\frac{|H_t|}{n} \mid H_{t - 1}}$ is bounded away from $0$). Thus, by using Chernoff bounds with $\varepsilon'$, we get w.h.p.\
    \[ (1 - \varepsilon')(T - 9 \cdot 2^{t - 1} \varepsilon') \leq \frac{|H_t|}{n} \leq (1 + \varepsilon') \left(T + 9 \cdot 2^{t - 1}\varepsilon' +  \frac{\varepsilon}{2} \right). \]
    Notice that
    \[ (1 - \varepsilon')(T - 9 \cdot 2^{t - 1} \varepsilon') \geq  T - 11 \cdot 2^{t - 1} \varepsilon' \geq T - \frac{11\varepsilon}{36} \geq \frac{1}{2} - \frac{15\varepsilon}{8},\]
    and
    \[ (1 + \varepsilon') \left(T + 9 \cdot 2^{t - 1}\varepsilon' + \frac{\varepsilon}{2} \right) \leq T + 13 \cdot 2^{t - 1}\varepsilon' + \frac{\varepsilon}{2} \leq T + \frac{13\varepsilon}{36} + \frac{\varepsilon}{2} \leq \frac{1}{2} - \frac{\varepsilon}{8}.\]
    which together imply the required lemma.
\end{proof}

In the rest of this section, we aim to show that $M_i$ does not change much over the course of the algorithm. First, we investigate the change in the value of $|M_i|$ for the first $t - 1$ iterations of the algorithm.

\begin{lemma}
    \label{lem:exp_mid}
    For $0 \leq i < t - 1$, suppose that $\frac{|H_i|}{n} \geq \frac{1}{2}$ and $\frac{|M_i|}{n} \geq \frac{|M_0|}{n}$, then $\ex{\frac{|M_{i + 1}|}{n} \mid M_i, H_i} \geq \frac{|M_i|}{n} \cdot 
 \left(1 + \frac{7\varepsilon}{4} \right)$.
\end{lemma}
\begin{proof}
    In the iteration $i$, for a given node $v$, consider the two nodes $u_1(v), u_2(v)$ that it picks in this iteration. Note that for $v$ to be in $M_{i + 1}$, it is sufficient that $u_1(v), u_2(v)$ are both in $M_i \cup H_i$ and at least one of them is in $M_i$ together with the messages sent out by $u_1(v), u_2(v)$ not being adversarially corrupted in this round. Thus, we have
    \begin{align*}
        \ex{\frac{|M_{i + 1}|}{n} \mid M_i, H_i} &\geq \frac{1}{n} \cdot \sum_{v \in V} \left( 2 \cdot \left(\frac{|M_i|}{n} - \beta \right) \cdot \left(\frac{|H_i|}{n} - \beta \right) + \left(\frac{|M_i|}{n} - \beta \right)^2 \right) \\
        &= \left(\frac{|M_i|}{n} - \beta \right) \left( \frac{2|H_i|}{n} + \frac{|M_i|}{n} - 3 \beta\right)\\
        &= \frac{|M_i|}{n} \cdot \left(1 - \frac{\beta n}{|M_i|} \right) \left( \frac{2|H_i|}{n} + \frac{|M_i|}{n} - 3 \beta\right).
    \end{align*}
    Now, we use our assumptions that $\frac{|H_i|}{n} \geq \frac{1}{2}$ and $\frac{|M_i|}{n} \geq \frac{|M_0|}{n} \geq 2 \varepsilon$ to obtain
    \begin{align*}
        \ex{\frac{|M_{i + 1}|}{n} \mid M_i, H_i} &\geq \frac{|M_i|}{n} \cdot \left(1 - \frac{\beta}{2 \varepsilon} \right) \cdot (1 + 2 \varepsilon - 3\beta) \\
        &\geq \frac{|M_i|}{n} \cdot \left(1 + 2 \varepsilon -4\beta - \frac{\beta }{2\varepsilon} \right) \\
        &\geq \frac{|M_i|}{n} \cdot \left(1 + \frac{7\varepsilon}{4} \right)
    \end{align*}
    where in the final step we have used the inequalities  $\beta \leq \frac{\varepsilon^{2.5}}{16} \leq \frac{\varepsilon^2}{16}$ and $\beta\leq\frac{\varepsilon}{200}$.
\end{proof}

Next, we investigate the final step of the algorithm.

\begin{lemma}
    \label{lem:exp_mid_fin}
    Suppose that $\frac{|H_{t - 1}|}{n} \geq \frac{1}{2} - \frac{3 \varepsilon}{2}$ and $\frac{|M_{t - 1}|}{n} \geq \frac{|M_0|}{n}$, then $\ex{\frac{|M_t|}{n} \mid M_{t - 1}, H_{t - 1}} \geq \frac{|M_{t - 1}|}{n}\cdot\left (1 - \frac{5\varepsilon}{4} \right)$.
\end{lemma}
\begin{proof}
    For the iteration $t$, we have 
    \begin{align*}
        &\ex{\frac{|M_t|}{n} \mid M_{t - 1}, H_{t - 1}} \\
        &\quad\geq \frac{1}{n} \cdot \sum_{v \in V} \biggl( \delta \left(2 \cdot \left(\frac{|M_{t - 1}|}{n} - \beta \right) \cdot \left(\frac{|H_{t - 1}|}{n} - \beta \right) + \left(\frac{|M_{t - 1}|}{n} - \beta \right)^2 \right) + (1 - \delta) \cdot \left(\frac{|M_{t - 1}|}{n} - \beta \right) \biggr) \\
        &\quad= \delta \cdot  \left(\frac{|M_{t - 1}|}{n} - \beta \right) \left( \frac{2|H_{t - 1}|}{n} + \frac{|M_{t - 1}|}{n} - 3 \beta\right) + (1 - \delta) \cdot \left(\frac{|M_{t - 1}|}{n} - \beta \right)\\
        &\quad= \frac{|M_{t - 1}|}{n} \cdot \left(1 - \frac{\beta n}{|M_{t - 1}|} \right) \left[\delta \cdot \left( \frac{2|H_{t - 1}|}{n} + \frac{|M_{t - 1}|}{n} - 3 \beta\right) + (1- \delta)\right].
    \end{align*}
    Now, we use our assumptions that $\frac{|H_{t - 1}|}{n} \geq \frac{1}{2} - \frac{3\varepsilon}{2}$ and $\frac{|M_{t - 1}|}{n} \geq \frac{|M_0|}{n} \geq 2 \varepsilon$ to obtain
    \begin{align*}
        \ex{\frac{|M_t|}{n} \mid M_{t - 1}, H_{t - 1}} &\geq \frac{|M_{t - 1}|}{n} \cdot \left(1 - \frac{\beta}{2 \varepsilon} \right) \cdot \left[\delta \cdot \left( 1 - 3\varepsilon + 2 \varepsilon - 3 \beta\right) + (1- \delta)\right]\\
        &\geq \frac{|M_{t - 1}|}{n} \cdot \left[1 - \frac{\beta}{2 \varepsilon} - \delta \cdot \left(\varepsilon + 3\beta \right)\right]\\
        &\geq \frac{|M_{t - 1}|}{n} \cdot \left(1 - \frac{5\varepsilon}{4} \right),
    \end{align*}
    where in the final step, we have used used the inequalities  $\beta \leq \frac{\varepsilon^{2.5}}{16} \leq \frac{\varepsilon^2}{16}$ and $\beta\leq\frac{\varepsilon}{200}$ and the fact that $\delta \leq 1$.
\end{proof}

Next, we prove concentration for every step.

\begin{lemma}
    \label{lem:con_mid}
    For $0 \leq i \leq t - 1$, suppose that $\frac{|M_i|}{n} \geq \frac{|M_0|}{n}$, then
    \[
    \pr{\left| \frac{|M_{i + 1}|}{n} - \ex{\frac{|M_{i + 1}|}{n} \mid M_i, H_i} \right| > \varepsilon' \cdot \ex{\frac{|M_{i + 1}|}{n} \mid M_i, H_i}} \leq e^{-\Omega(\varepsilon' \ex{|M_{i + 1}| \mid M_i, H_i})}.
    \]
\end{lemma}
\begin{proof}
    For any $v$, notice that the event $v \in M_{i + 1}$ is conditionally independent given $M_i, H_i$. The required claim then immediately follows from Chernoff bounds.
\end{proof}

Finally, we are ready to bound the size of $|M_t|$.

\begin{lemma}
    \label{lem:fin_mid}
    Let $\varepsilon = \Omega \left(\left(\frac{\log n}{n} \right)^{1/5} \right)$. At the end of the algorithm, w.h.p.\ $\frac{|M_t|}{n} \geq 2\varepsilon$.
\end{lemma}
\begin{proof}
    First, note that for $0 \leq i \leq t - 2$, w.h.p., we must have $\frac{|H_i|}{n} > 1/2$. Otherwise, by Lemma \ref{lem:exp2} and using Chernoff bounds with error $\varepsilon' = 1/100$, and the inequalities $\varepsilon \leq \frac{1}{6}, \beta \leq \frac{\varepsilon}{200} \leq \frac{1}{100}$, we have w.h.p.\
    \[\frac{|H_{i + 1}|}{n} \leq (1 + \varepsilon') \left(\frac{1}{2} + \beta \right)^{2} \leq \frac{101}{100} \cdot \left(\frac{1}{2} + \frac{1}{100} \right)^{2} \leq \frac{9}{32} \leq \frac{1}{2} - \frac{21 \varepsilon}{16}.\]
    This implies that $t = i + 1$, which is a contradiction.

    Next, we show by induction that w.h.p.\ $\frac{|M_i|}{n} \geq \frac{|M_0|}{n}$ for $0 \leq i \leq t - 1$. The claim is trivially true for $i = 0$. Otherwise, suppose that $\frac{|M_i|}{n} \geq \frac{|M_0|}{n}$ is true for some $i$. Using Lemma \ref{lem:exp_mid} and \ref{lem:con_mid} with $\varepsilon' \coloneqq \varepsilon / 4$, we have
    \begin{align*}    
    &\pr{\left| \frac{|M_{i + 1}|}{n} - \ex{\frac{|M_{i + 1}|}{n} \mid M_i, H_i} \right| > \varepsilon' \cdot \ex{\frac{|M_{i + 1}|}{n} \mid M_i, H_i}} \\
    &\qquad\leq e^{-\Omega(\varepsilon' \ex{|M_{i + 1}| \mid M_i, H_i})} = e^{-\Omega((\varepsilon')^2 \varepsilon n)} = n^{- \omega(1)}.
    \end{align*}
    Thus, w.h.p.\ we have 
    \[\frac{|M_{i + 1}|}{n} \geq (1 - \varepsilon') \cdot \ex{\frac{|M_{i + 1}|}{n} \mid M_i, H_i} \geq \left(1 - \frac{\varepsilon}{100}\right) \cdot \left(1 + \frac{7\varepsilon}{4} \right) \cdot \frac{|M_i|}{n} \geq \frac{|M_i|}{n},\]
    where in the penultimate step we have used Lemma \ref{lem:exp_mid}. 
    
    Next, using Lemma \ref{lem:cumerr} w.h.p., we have that w.h.p.\ 
    \[ \frac{|H_{t - 1}|}{n} \geq (1 - \tfrac{\varepsilon}{9 \cdot 2^{t + 1}})^{2^{t} - 1}h'_{t - 1}  \geq (1 - (2^{t} - 1)\tfrac{\varepsilon}{9 \cdot 2^{t + 1}})h'_{t - 1} \geq \left(1 - \frac{\varepsilon}{18} \right)T \geq T - \frac{\varepsilon}{18} \geq \frac{1}{2} - \frac{3\varepsilon}{2}.\]
    Hence the conditions of Lemma \ref{lem:exp_mid_fin} are satisfied, and we can use it together with Lemma \ref{lem:con_mid} to get that w.h.p.\
    \[ \frac{|M_t|}{n} \geq (1 - \varepsilon') \cdot \ex{\frac{|M_{t}|}{n} \mid M_{t-1}, H_{t-1}} \geq \left(1 - \frac{\varepsilon}{100} \right) \cdot \left(1 - \frac{5\varepsilon}{4} \right) \cdot \frac{|M_{t - 1}|}{n}.\]
    Furthermore, using Lemma \ref{lem:exp_mid} together with Lemma \ref{lem:con_mid}, we have w.h.p.\
    \begin{align*}
    \frac{|M_t|}{n} &
    \geq \left(1 - \frac{\varepsilon}{100} \right)^2 \cdot \left(1 - \frac{5\varepsilon}{4} \right) \cdot \left(1 + \frac{7\varepsilon}{4} \right) \cdot \frac{|M_{t - 2}|}{n}
    \geq \left(1 - \frac{127\varepsilon}{100} \right) \cdot \left(1 + \frac{7\varepsilon}{4} \right) \cdot \frac{|M_0|}{n} \\
    &= \left(1 + \frac{48\varepsilon}{100}-\frac{889\varepsilon^2}{400} \right) \cdot \frac{|M_0|}{n}
    \geq \frac{|M_0|}{n} = 2 \varepsilon,
    \end{align*}
    where in the penultimate step, we have used the assumption $\varepsilon \leq \frac{1}{6}$.
\end{proof}

Finally, we can prove the following lemma which implies fast quantile shifting.

\begin{lemma}
    Let $\varepsilon = \Omega \left(\left(\frac{\log n}{n} \right)^{1/5} \right)$. At the end of iteration $t$ of Algorithm \ref{alg:2tourn}, w.h.p.\ any $\phi'$-quantile where $\phi' \in \left[ \frac{1}{2} - \frac{\varepsilon}{8}, \frac{1}{2} + \frac{\varepsilon}{8} \right]$ must be in $M_t$.
\end{lemma}
\begin{proof}
    By Lemmas \ref{lem:Hbound} and \ref{lem:fin_mid}, we have $\frac{1}{2} - \frac{15 \varepsilon}{8} \leq \frac{|H_t|}{n} \leq \frac{1}{2} - \frac{\varepsilon}{8}$ and $\frac{|M_t|}{n} \geq 2 \varepsilon$ w.h.p. Consequently, we must have $\frac{|H_t|}{n} + \frac{|M_t|}{n} \geq \frac{1}{2} + \frac{\varepsilon}{8}$. Thus, we conclude that any $\phi'$-quantile where $\phi' \in \left[ \frac{1}{2} - \frac{\varepsilon}{8}, \frac{1}{2} + \frac{\varepsilon}{8} \right]$ must be in $M_t$.
\end{proof}

\section{Approximate Mean}\label{sec:approx_mean_appendix}

In this section, we prove the runtime and correctness guarantees for our adversarially robust approximate mean algorithm. We start by recalling Theorem~\ref{thm:mean}.

\mean*

We remark that by assumption $\beta\leq \left(\frac{\varepsilon}{100}\right)^{2.5}$ and $\delta << \left(\frac{\varepsilon}{100}\right)^{2.5}$, which implies that $\eta \leq \left(\frac{\varepsilon}{100}\right)^{2.5}$, where $\delta$ is defined on line $1$ of Algorithm \ref{alg:pullavg}. Furthermore, we can rewrite $\eta$ as 
\begin{align*}
   \eta &\coloneqq \max(\beta, \delta, \min(\gamma^5, (\tfrac{\varepsilon}{100})^{2.5}))\\
   &= \max(\min(\beta, (\tfrac{\varepsilon}{100})^{2.5}), \min(\delta, (\tfrac{\varepsilon}{100})^{2.5}), \min(\gamma^5, (\tfrac{\varepsilon}{100})^{2.5}))\\
   &= \min(\max(\beta, \delta, \gamma^5), (\tfrac{\varepsilon}{100})^{2.5}).
\end{align*} 

Recall the definitions of $\psi(t) \coloneqq \sum_{u \in V} x_u(t)$ and $\Phi(t) \coloneqq \sum_{u \neq v \in V} (x_u(t) - x_v(t))^2$.
We first give two alternate expressions for $\Phi(t)$ in terms of $\psi(t)$.

\begin{lemma}
    \label{lem:alt}
    \[ \Phi(t) = \left(n \cdot \sum_{u \in V} x_u(t)^2\right) - \psi(t)^2 = n\cdot \sum_{u \in V} \left( x_u(t) - \frac{\psi(t)}{n} \right)^2.\]
\end{lemma}
\begin{proof}
    First, by definition of $\psi(t)$ we have
    \[ \psi(t)^2 = \left( \sum_{u \in V} x_u(t) \right)^2 = \sum_{u \in V} \sum_{v \in V} x_u(t)x_v(t).\]
    Next, 
    \begin{align*}
       \Phi(t) &= \sum_{u\neq v \in V} (x_u(t) - x_v(t))^2 \\
       &= \frac{1}{2} \sum_{u \in V} \sum_{v \in V} (x_u(t) - x_v(t))^2\\
        &= \frac{1}{2} \sum_{u \in V} \sum_{v \in V} (x_u(t)^2 + x_v(t)^2 - 2x_u(t)x_v(t)\\
        &= \left(n \cdot \sum_{w \in V} x_w(t)^2\right) - \left( \sum_{u \in V} \sum_{v \in V} x_u(t)x_v(t) \right)\\
        &= \left(n \cdot \sum_{w \in V} x_w(t)^2\right) - \psi(t)^2.
    \end{align*}
    Furthermore, 
    \begin{align*}
        n\cdot \sum_{u \in V} \left( x_u(t) - \frac{\psi(t)}{n} \right)^2 &= n\cdot \sum_{u \in V} \left( x_u(t)^2 + \frac{\psi(t)^2}{n^2} - 2x_u(t) \cdot \frac{\psi(t)}{n} \right) \\
        &= \left(n \cdot \sum_{u \in V} x_u(t)^2\right) + \psi(t)^2 - 2 \left( \sum_{u \in V} x_u(t) \right) \cdot \psi(t) \\
        &= \left(n \cdot \sum_{u \in V} x_u(t)^2\right) - \psi(t)^2.
    \end{align*}
\end{proof}

We now come to the analysis of the algorithm. First, let us compute the expected value of a single term in the potential, and show that the updated value of a node $u$ is close (subject to the precision constraints imposed by the fraction of messages corrupted and the bound on the values $M$) to the current average of values $\frac{\psi(t)}{n}$. 

\begin{lemma}
    \label{lem:expval}
    For arbitrary $u, v \in V$ such that $u \neq v$, and for all $0 \leq t \leq T$, we have
    \begin{align}\label{eq:expected-difference-bound}
        \ex{(x_u(t+1) - x_v(t+1))^2 \mid \Phi(t)} \leq \frac{\Phi(t)}{n^2} + 4 \beta M^2
    \end{align}
    and
    \begin{align}\label{eq:expected-mean-bounds}
        \frac{\psi(t)}{n} - 2 \beta M \leq \ex{x_u(t+1) \mid \psi(t)} \leq \frac{\psi(t)}{n} + 2 \beta M.
    \end{align}
\end{lemma}
\begin{proof}
    Denote by $S_t \coloneqq \{ x_{w}(t) \mid w \in V\}$ the set of values stored inn the nodes at iteration $t$.
    Let $u_1, u_2$ be the random nodes chosen by $u$ in this round, and similarly let $v_1, v_2$ be the random nodes chosen by $v$ in this round. Let $F_u$ denote the event that $u_1, u_2$ are both \emph{not} adversarially affected in this round, and similarly denote by $F_v$ the event that $v_1, v_2$ are not adversarially affected (in the same round). By the union bound, we have $\pr{\overline{F_u}}, \pr{\overline{F_v}} \le 2\beta$ and $\pr{\overline{F_u \cap F_v}} \leq 4 \beta$.
   Now, notice that
    \begin{align*}
       (x_u(t+1) - x_v(t+1))^2 &= \left(\frac{x_{u_1}(t) + x_{u_2}(t)}{2} - \frac{x_{v_1}(t) + x_{v_2}(t)}{2} \right)^2 \\
       &= \frac{1}{4} \cdot [x_{u_1}(t)^2 + x_{u_2}(t)^2 + x_{v_1}(t)^2 + x_{v_2}(t)^2] \\
       &\qquad + \frac{1}{2} \cdot [x_{u_1}(t)x_{u_2}(t) + x_{v_1}(t)x_{v_2}(t) \\
       &\qquad \qquad - x_{u_1}(t)x_{v_1}(t) - x_{u_1}(t)x_{v_2}(t) \\
       &\qquad \qquad - x_{u_2}(t)x_{v_1}(t) - x_{u_2}(t)x_{v_2}(t)]
    \end{align*}
    Observe that conditioned on $F_u \cap F_v$, all of $x_{u_1}(t), x_{u_2}(t), x_{v_1}(t), x_{v_2}(t)$ are independently and identically distributed. Thus, we get
    \[ \ex{(x_{u}(t+1) - x_{v}(t+1))^2 \mid S_t, F_u \cap F_v} =  \ex{x_{u_1}(t)^2 \mid S_t, F_u \cap F_v} - \ex{x_{u_1}(t) \mid S_t, F_u \cap F_v}^2.\]
    We can now compute the expectations that appear on the right hand side explicitly.
    \[ \ex{x_{u_1}(t)^2 \mid S_t, F_u \cap F_v} = \sum_{w \in V} \frac{1}{n} \cdot x_{w}(t)^2\]
    \[ \ex{x_{u_1}(t) \mid S_t, F_u \cap F_v} = \sum_{w \in V} \frac{1}{n} \cdot x_{w}(t) = \frac{\psi(t)}{n}\]
    Substituting in the previous two equations and using Lemma \ref{lem:alt},
    \[ \ex{(x_{u}(t+1) - x_{v}(t+1))^2 \mid S_t, F_u \cap F_v} = \sum_{w \in V} \frac{1}{n} \cdot x_{w}(t)^2 - \left( \frac{\psi(t)}{n} \right)^2 = \frac{\Phi(t)}{n^2}.\]
    We can then show \eqref{eq:expected-difference-bound}: 
    \begin{align*}
        \ex{(x_{u}(t+1) - x_{v}(t+1))^2 \mid S_t} &= \pr{F_u \cap F_v} \cdot \ex{(x_{u}(t+1) - x_{v}(t+1))^2 \mid S_t, F_u \cap F_v} \\
        &\qquad + \pr{\overline{F_u \cap F_v}} \cdot \ex{(x_{u}(t+1) - x_{v}(t+1))^2 \mid S_t, \overline{F_u \cap F_v}} \\
        &\leq \pr{F_u \cap F_v} \cdot \frac{\Phi(t)}{n^2} + \pr{\overline{F_u \cap F_v}} \cdot (M)^2 \\
        &\leq  \frac{\Phi(t)}{n^2} + 4 \beta M^2,
    \end{align*}
    where we have used the fact that $(x_{u}(t+1) - x_{v}(t+1))^2 \leq M^2$ since we clip the values to the interval $[0,M]$. Finally, notice that the expectation bound depends only on the value $\Phi(t)$ from the iteration $t$. This finishes the proof of \eqref{eq:expected-difference-bound}.

    For \eqref{eq:expected-mean-bounds}, notice that
    \begin{align*}
        \ex{x_{u}(t+1) \mid S_t} &= \pr{F_u} \cdot \ex{x_{u}(t+1)\mid S_t, F_u} + \pr{\overline{F_u}} \cdot \ex{x_{u}(t+1) \mid S_t, \overline{F_u}} \\
        &\geq \pr{F_u} \cdot \ex{x_{u}(t+1)\mid S_t, F_u} \\
        &\geq  (1 - 2 \beta)\frac{\psi(t)}{n} \\
        &\geq \frac{\psi(t)}{n} - 2 \beta M
    \end{align*}
    where we have used the fact that $x_{u}(t+1) \geq 0$ and $\psi(t) \leq n M$. Furthermore,
    \begin{align*}
        \ex{x_{u}(t+1) \mid S_t} &= \pr{F_u} \cdot \ex{x_{u}(t+1)\mid S_t, F_u} + \pr{\overline{F_u}} \cdot \ex{x_{u}(t+1) \mid S_t, \overline{F_u}} \\
        &\leq \ex{x_{u}(t+1)\mid S_t, F_u} + 2 \beta \cdot \ex{x_{u}(t+1) \mid S_t, \overline{F_u}} \\
        &\leq  \frac{\psi(t)}{n} + 2 \beta M
    \end{align*}
    where we have used the fact that $x_{u}(t+1) \leq M$. Again, notice that the expectation bounds depend only on the value $\psi(t)$ from the iteration $t$. Thus, we get \eqref{eq:expected-mean-bounds}.
\end{proof}
The next lemma provides a w.h.p.\ upper bound on the potential $\Phi(t)$ at the end of Phase 1 of the algorithm.

\begin{lemma}
    \label{lem:potdec}
    $\Phi(T) \leq 5 \eta n^2 M^2$ with high probability.
\end{lemma}
\begin{proof}
    First, notice that for all $0 \leq t < T$, we have by Lemma \ref{lem:expval}
    \begin{align*}
        \ex{\Phi(t + 1) \mid \Phi(t)} &= \sum_{u\neq v \in V} \ex{(x_{u}(t+1) - x_{v}(t+1))^2 \mid \Phi(t)} \\
        &\leq \binom{n}{2} \left(\frac{\Phi(t)}{n^2} + 4 \beta M^2 \right) \\
        &\leq \frac{\Phi(t)}{2} + 2 \beta n^2 M^2.
    \end{align*}
    Now, notice that $\Phi(t + 1)$ is a function of $2n$ independent random variables that correspond to the $2$ nodes chosen by each of the $n$ nodes in this round. Each of these variables appear in exactly $n - 1$ of the summation terms and can affect this term by at most $M^2$. Thus, the effect of each of these variables on $\Phi(t + 1)$ is bounded by $n M^2$, so the sum of squares of the effects of all these random variables is bounded by $2n^3 M^4$. 
    
    First, assume that $\Phi(t) > 5 \eta n^2 M^2 \geq 5 \delta n^2 M^2 = \Omega( M^2 n^{3/2}(\log n)^{1/2})$, where we have used the fact that by definition $\delta = \frac{(\log n)^{1/2}}{n^{1/2}}$ and that $\eta \geq \delta$. We can now apply Azuma's inequality to obtain
    \[ \pr{\Phi(t + 1) - \ex{\Phi(t + 1) \mid \Phi(t)} \geq \frac{\Phi(t)}{18}} \leq 2\exp \left( - \frac{\Phi(t)^2}{18^2 \cdot 2 \cdot (2n^3M^4)} \right) = 2e^{-\Omega(\log n)} = n^{-\Omega(1)}.\]
    Thus, by a union bound, the following holds for all iterations with $\Phi(t) > 5 \eta n^2 M^2$ w.h.p.\
    \[ \Phi(t + 1) \leq \frac{5\Phi(t)}{9} + 2 \beta n^2 M^2 \leq \frac{5\Phi(t)}{9} + 2 \eta n^2 M^2.\]

    Now, suppose that $\Phi(t) \leq 5 \eta n^2 M^2$. Then, again by applying Azuma's inequality, we have
    \[\pr{\Phi(t + 1) - \ex{\Phi(t + 1) \mid \Phi(t)} \geq \frac{\eta n^2 M^2}{2}} \leq 2\exp \left( - \frac{\eta^2 n^4 M^4}{2^2 \cdot 2 \cdot (2n^3M^4)} \right) = 2e^{-\Omega(\log n)} = n^{-\Omega(1)},\]
    where we have used that $\eta^2 \ge \delta^2 = \frac{\log n}{n}$.
    Thus, by a union bound, the following holds for all iterations with $\Phi(t) \leq 5 \eta n^2 M^2$ w.h.p
    \[ \Phi(t + 1) \leq \frac{\Phi(t)}{2} + 2 \beta n^2 M^2 + \frac{\eta n^2 M^2}{2} \leq 5 \eta n^2 M^2. \]
    
    Thus, putting it all together with a union bound, we have
    \[ \Phi(t + 1) \leq \Phi(0) \cdot \left(\frac{5}{9} \right)^T + \frac{2 \eta n^2 M^2}{1 - \frac{5}{9}} \leq 5 \eta n^2 M^2,\]
    where we have used the fact that $\Phi(0) \leq \binom{n}{2} (M)^2 \leq \frac{n^2 M^2}{2}$, and that $\left(\frac{5}{9} \right)^T \leq \eta$.
\end{proof}

Next, we will show that $\psi(T)$ does not deviate much from $\psi(0)$.

\begin{lemma}
    \label{lem:valuesum}
    $|\psi(T) - \psi(0)| \leq \frac{3\varepsilon}{8} n M$ with high probability.
\end{lemma}
\begin{proof}
    First, notice that 
    \[\beta T = \beta \left \lceil \log_{9/5}\frac{1}{\eta} \right \rceil \leq (3.5)\beta \log \frac{1}{\eta} \leq (3.5) \beta \log \frac{1}{\beta} \leq (3.5) \beta \left(\frac{1}{\beta} \right)^{0.2} \leq (3.5) \left( \frac{\varepsilon}{100} \right)^2 \leq \frac{\varepsilon}{100},\] 
    where the fourth inequality holds because $\frac{1}{\beta} \ge \left(\frac{100}{\varepsilon}\right)^{2.5} \ge 200^{2.5}$ and $\log x \le x^{0.2}$ for all $x\ge 200^{2.5}$, and in the penultimate step, we have used the fact that $\beta \leq \left(\frac{\varepsilon}{100}\right)^{2.5}$. 
    
    Thus, for all $0 \leq t < T$, using Lemma \ref{lem:expval} we have
    \[\ex{\psi(t + 1) \mid \psi(t)} = \sum_{u \in V} \ex{x_{u}(t+1) \mid \psi(t)} \leq n \left(\frac{\psi(t)}{n} + 2\beta M \right) \leq \psi(t) + \frac{\varepsilon n M}{4T}.\]
    Similarly, we also get that 
    \[\ex{\psi(t + 1) \mid \psi(t)} = \sum_{u \in V} \ex{x_{u}(t+1) \mid \psi(t)} \geq n \left(\frac{\psi(t)}{n} - 2\beta M \right) \geq \psi(t) -\frac{\varepsilon n M}{4T}.\] 
    Now, notice that $\psi(t + 1)$ is a function of $2n$ independent random variables that correspond to the $2$ nodes chosen by each of the $n$ nodes in this round. Each of these variables can affect $\psi(t+1)$ by at most $\frac{M}{2}$. Thus, the sum of squares of the effects of all these random variables on $\psi(t + 1)$ is bounded by $\frac{n M^2}{2}$. Under our assumption that $\varepsilon = \Omega \left(\frac{(\log n)^{6/5}}{n^{1/5}}\right)$, we can deduce that $\frac{\varepsilon}{T} = \frac{\varepsilon}{O(\log n)} = \Omega \left(\frac{(\log n)^{1/5}}{n^{1/5}}\right) = \Omega \left(\frac{(\log n)^{1/2}}{n^{1/2}}\right)$, and apply Azuma's inequality to obtain
    \[ \pr{|\psi(t + 1) - \ex{\psi(t + 1) \mid \psi(t)}| \geq \frac{\varepsilon n M}{8T}} \leq 2 \cdot \exp \left( - \frac{(\varepsilon n M)^2}{8^2 T^2 \cdot 2 \cdot \frac{n M^2}{2}} \right) \leq 2 e^{-\Omega(\log n)} = n^{-\Omega(1)}.\]
    Thus, by a union bound, the following holds for all iterations w.h.p.\
    \[ \psi(t) - \frac{3\varepsilon n M}{8T} \leq \psi(t + 1) \leq \psi(t) + \frac{3\varepsilon n M}{8T}.\]
    Thus, putting it all together with a union bound, we may bound the total deviation of $\psi(T)$ from the initial value $\psi(0)$ by
    \[ \psi(0) - \frac{3\varepsilon n M}{8} \leq \psi(T) \leq \psi(0) + \frac{3\varepsilon n M}{8},\]
    which gives us the desired result.
\end{proof}

Now, we are ready to show that after $T$ iterations, most of the nodes are close to the true value.

\begin{lemma}
    \label{lem:iterend}
    At the end of $T$ iterations, w.h.p.\ at most a $\frac{\eta^{0.2}}{2}$ fraction of the nodes deviate from the true average (namely, $\frac{\psi(0)}{n}$) by more than $\varepsilon M$.
\end{lemma}
\begin{proof}
    Notice that by Lemma \ref{lem:valuesum}, w.h.p.\ $|\psi(T) - \psi(0)| \leq \frac{3\varepsilon}{8} n M$. Hence, it suffices to show that at most a $\frac{\eta^{0.2}}{2}$ fraction of the nodes deviate from $\frac{\psi(T)}{n}$ by more than $\frac{5\varepsilon}{8} M$.
    
    Assume the contrary. Then, using Lemma \ref{lem:alt} we have
    \[ \Phi(T) = n\cdot \sum_{u \in V} \left( x_{u}(T) - \frac{\psi(T)}{n} \right)^2  > n \cdot \frac{\eta^{0.2}}{2} n \cdot \left( \frac{5}{8} \right)^2 \varepsilon^2 M^2 \geq 5 \eta n^2 M^2\]
    where we have used the fact that $\eta \leq \left(\frac{\varepsilon}{100}\right)^{2.5}$. But, this contradicts Lemma \ref{lem:potdec}. Hence, our assumption must be wrong and the claim must be true.
\end{proof}

Finally, we analyze the second and final phase of the algorithm.

\begin{lemma}\label{lem:final-count}
    W.h.p.\ all but a $\gamma$ fraction of the nodes outputs an average in $\left[ \frac{\psi(0)}{n} - \varepsilon M, \frac{\psi(0)}{n} + \varepsilon M\right]$
\end{lemma}
\begin{proof}
    By Lemma \ref{lem:iterend}, we know that after $T$ iterations, at most a $\frac{\eta^{0.2}}{2}$ fraction of the nodes hold a value outside the desired interval.

    First, suppose that $\eta = \min(\gamma^5, \left(\frac{\varepsilon}{100}\right)^{2.5})$. Then, notice that $\frac{\eta^{0.2}}{2} \leq \frac{\gamma}{2} \leq \gamma$, and hence the required claim follows immediately, as Phase $2$ does nothing in this case. 

    Now, suppose that $\eta = \delta$. Now, a fixed node $v$ outputs a value in the required range if at least $K/2$ of the nodes that it pulls from in the final part of the algorithm has an approximation in the desired range. This means that the probability that this this does not happen is at most 
    \[ \binom{K}{K/2} \left(\frac{\delta^{0.2}}{2}\right)^{K/2} \leq (2 \delta^{0.2})^{K/2} \leq (2 \delta^{0.2})^{50} = \frac{2^{50}\log^5 n}{n^5} \le \left(\frac{1}{2n}\right)^4 \leq \left(\frac{\gamma}{2}\right)^4. \]
    
     Now, suppose that $\eta = \beta$. By the same reasoning as above, the probability that a fixed node does not output a value in the desired range is upper bounded by
    \[ \binom{K}{K/2} \left(\frac{\beta^{0.2}}{2}\right)^{K/2} \leq  (2\beta^{0.2})^{K/2} = (32\beta)^{K/10} \leq \left(\frac{\gamma}{2}\right)^4, \]
    where the last inequality comes from the fact that $K \ge \ceil{40 \cdot \log_{32\beta} \left( \frac{\gamma}{2} \right)}$. 

    So whenever $\eta=\delta$ or $\eta=\beta$, the expected number of nodes that do not output a correct approximate median is at most $\frac{\gamma^4 n}{16}$.
    First, suppose that $\gamma = \Omega \left( \frac{\log n}{n}\right)$. Notice that the expected number of nodes is $\frac{\gamma^4 n}{16} \leq \frac{\gamma n}{2}$. Now, using Chernoff bounds, the probability that the number of nodes that output the wrong answer is more than $\gamma n$ is upper bounded by $\exp(-\frac{1}{3} \cdot  (\frac{\gamma n}{2})) = \exp(- \Omega(\log n)) = n^{-\Omega(1)}$.

    Else, we have $\gamma = o \left( \frac{\log n}{n}\right)$. Now, we can apply Markov's inequality to conclude that the probability that more than $\gamma n$ nodes output the wrong answer is upper bounded by $\frac{\gamma^4 n}{16 \cdot \gamma n} = o \left( \frac{\log^3 n}{n^3} \right) = n^{-\Omega(1)}$.
\end{proof}
This concludes the correctness proof of our algorithm. We tie up this section by proving the claimed runtime bound.
\begin{theorem}
    \label{thm:run-time}
    The  algorithm terminates in $O \left(\log \frac{1}{\varepsilon} + \log \frac{1}{\gamma + \beta} + \frac{\log \gamma}{\log \beta} \right)$ rounds.
\end{theorem}
\begin{proof}
    Notice that the total number of rounds is $O(T + K) = O \left( \log \frac{1}{\eta} + \frac{\log \gamma}{ \log \beta} \right)$.  Now, recall that $\eta$ can be alternatively written as
    \[ \eta = \min(\max(\beta, \delta, \gamma^5), \left(\tfrac{\varepsilon}{100}\right)^{2.5}) \geq \min(\max(\beta, \gamma^5), \left(\tfrac{\varepsilon}{100}\right)^{2.5}).\] 
    Thus, one can rewrite the time complexity as 
    \begin{align*}
       O \left(\log \frac{1}{\min(\max(\beta, \gamma^5), \left(\frac{\varepsilon}{100}\right)^{2.5})} + \frac{\log \gamma}{\log \beta} \right) &= O \left( \max \left(\log \frac{1}{\varepsilon}, \min \left(\log \frac{1}{\gamma}, \log \frac{1}{\beta} \right)\right) + \frac{\log \gamma}{\log \beta} \right) \\
       &= O \left(\log \frac{1}{\varepsilon} + \min \left(\log \frac{1}{\gamma}, \log \frac{1}{\beta} \right) + \frac{\log \gamma}{\log \beta} \right) \\
       &= O \left(\log \frac{1}{\varepsilon} + \log \frac{1}{\gamma + \beta} + \frac{\log \gamma}{\log \beta} \right).
    \end{align*}
\end{proof}

\section{Lower Bounds on Round Complexity}\label{sec:low_bound_appendix}

This section provides the proof of Proposition \ref{prop:lower-bound}.

\lowerboundprop*

\begin{proof}
    We first note that the addition of adversarial faults can only make the problem harder, in the sense that the adversaries can choose to not corrupt any messages and the lower bound in this scenario should still be a lower bound for the original problem. Thus, we will ignore the presence of adversarial faults in the following analysis.

    Consider the following two scenarios. The first is when each node is associated with a distinct value from $\{1, 2, \ldots, n\}$, and the second is when each node is associated with a distinct value from $\{1 + \lfloor 3 \varepsilon n \rfloor, \ldots, n + \lfloor 3 \varepsilon n \rfloor\}$. Notice that $\left \lfloor \frac{n + 1}{2} \right \rfloor$ is a median for the first scenario and $\left \lfloor \frac{n + 1}{2} \right \rfloor + \lfloor 3 \varepsilon n \rfloor$ is a median for the second scenario. Similarly, $\frac{n + 1}{2}$ and $\frac{n + 1}{2} + \lfloor 3 \varepsilon n \rfloor$ are the means of the first and second scenarios respectively. Thus, the mean (median) for one of these scenarios is not a correct approximation for the mean (median) of the other scenario.
    
    Without loss of generality we can assume that $n-\lfloor 3 \varepsilon n \rfloor$ nodes start with the same values in the two scenarios, and we denote by $S$ the set of $\lfloor 3 \varepsilon n \rfloor$ nodes that have a different starting value depending on the scenario. 
    A node can distinguish between the two scenarios only if it receives a value from the set $S$. 
    Let us call a node \textit{good} if it has received a value from $S$ (either directly or indirectly), and \textit{bad} otherwise. Note that a bad node cannot output a correct answer with probability better than $\frac{1}{2}$ (random guessing). Initially there are at most $\lfloor 3 \varepsilon n \rfloor \leq 3 \varepsilon n$ good nodes.

    Let $X_i$ denote the number of good nodes at the end of round $i$, with $X_0 = \lfloor 3 \varepsilon n \rfloor$. Given a bad node $v$, it can become a good if it pulls from a good node, or if it is pushed to from a good node. Let $Y_v$ denote the event that $v$ pulls from a good node. We have $\pr{Y_v \mid X_i} = X_i / n$. Also, the pushes from the good nodes can only generate another at most $X_i$ good nodes. Therefore,
    \[ \ex{X_{i + 1} \mid X_i} \leq 2X_i + \sum_{v\text{ bad}}\ex{Y_v \mid X_i} \leq 3X_i.\]
    Let $B_i$ denote the set of bad nodes at the end of round $i$. Since all pulls and pushes are independent, we have by Chernoff $\pr{\sum_{v \in B_i} Y_v > 2X_i \mid X_i} \leq e^{-X_i / 3} \leq e^{-5 \log n} = 1 / n^5$, since $X_i \geq X_0 \geq 15 \log n$. Therefore, with probability at least $1 - 1 / n^5$, $X_{i + 1} \leq 4X_i$. By taking a union bound over such events for the first $t' = \log _4 (1 / 6\varepsilon)$ rounds, we conclude with probability at least $1 - 1/n^4$, $X_{t'} \leq (3 \varepsilon n) 4^{t'} \leq n / 2$.

    Let $t_0$ be the last round such that $X_{t_0} \leq n / 2$. Define $Z_i \coloneqq |B_i| / n$. A node $v$ remains in $B_i$ if it did not pull from a good node and it was not pushed to from any good node. Denoting such an event by $W_v$, we have $\pr{W_v \mid B_i} \geq Z_i \cdot \left(1 - \frac{1}{n} \right)^{n - 1} \geq Z_i \cdot e^{-1}$. In the second inequality, we used that for all $n \ge 1$ , it holds that  $\left(1 -\frac{1}{n} \right)^{n - 1} \geq e^{-1}$ (Corollary 1.4.6. in~\cite{doerr2020probabilistic}).
    \[ \ex{|B_{i + 1}| \mid B_i} = \ex{ \sum_{v \in B_i} W_v \mid B_i} \geq \left( \sum_{v \in B_i} Z_i \cdot e^{-1} \right) = |B_i| \cdot Z_i \cdot e^{-1}. \]

    Note that the events $\{ W_v \}_{v \in B_i}$ are negatively associated, which means that Chernoff bounds hold for their sum. Suppose that $Z_i \cdot |B_i| \geq 40 e \log n$. By Chernoff, 
    \[ \pr{|B_{i + 1}| \leq |B_i| \cdot Z_i \cdot e^{-1} / 2 \mid B_i} \leq e^{-Z_i \cdot |B_i| \cdot e^{-1} / 8} \leq 1 / n^5. \]

    In other words, $\pr{Z_{i + 1} \leq Z_i^2 / (2e) \mid Z_i} \leq 1/n^5$. Suppose that $t_1 = O(n)$; then, we can take a union bound over the subsequent $t_1$ rounds (after the number of good nodes reaches $n/2$, i.e.\ after $t_0$ rounds) to show that with probability at least $1 - O(1/n^4)$, $Z_{i + 1} \geq Z_i^2 / (2e)$ holds for these rounds, as long as $Z^2_{t_0 + t_1 - 1} \geq 40 e \log n / n$.

    Let $z_0 \coloneqq 1/2$ and $z_{i + 1} \coloneqq z_i^2 / (2e)$. If $z_{t_1} \geq 20 \log n / n$, we have $Z_{t_0 + t_1 - 1} \geq z_{t_1 - 1}$ with probability at least $1 - O(1/n^4)$. Let $t_1 \coloneqq \left \lfloor \log_2 \log_{4e} \left(\frac{1}{\gamma} \right) \right \rfloor = O(n)$; then by definition of $z_{t_1}$, we have 
    \[z_{t_1} = \left( \frac{1}{2e}\right)^{2^{t_1} - 1} \cdot z_{0}^{2^{t_1}} \geq \left( \frac{1}{4e}\right)^{2^{t_1}} \geq \gamma.\]
    Notice that $\gamma \geq \frac{20 \log n}{n}$ by assumption, which means that $z_{t_1} \geq \frac{20 \log n}{n}$, and hence $Z^2_{t_0 + t_1 - 1} \geq z^2_{t_1 - 1} = (2e) z_{t_1} \geq \frac{40e \log n}{n}$ with probability at least $1 - O(1/n^4)$.

    Recall that $t_0 \geq \log_4 \frac{1}{6 \varepsilon}$ with probability at least $1 - 1 / n^4$. Thus, a union bound over the two events implies that with probability at least $1 - O(1/n^4)$, at least a $\gamma$ fraction of the nodes are bad after $t_0 + t_1$ rounds.

    Therefore, the probability that all but a $\gamma$ fraction of the nodes compute the correct output at the end of round $t_0 + t_1$ is at most $1/2 + O(1/n^4)$, which proves our required claim.
\end{proof}

\end{document}